\documentclass[12pt,notitlepage]{amsart}%
\usepackage{amssymb}
\usepackage{amsfonts}
\usepackage{graphicx}
\usepackage{amscd}
\usepackage{graphicx}
\usepackage{amsmath}
\usepackage{amsxtra}
\usepackage{footmisc}%
\newtheorem{theorem}{Theorem}
\theoremstyle{plain}

\newtheorem{definition}{Definition}

\newtheorem{lemma}{Lemma}
\newtheorem{notation}{Notation}

\newtheorem{proposition}{Proposition}
\newtheorem{remark}{Remark}

\numberwithin{equation}{section}
\numberwithin{theorem}{section}
\numberwithin{lemma}{section}
\numberwithin{proposition}{section}
\numberwithin{corollary}{section}

\ifx\pdfoutput\relax\let\pdfoutput=\undefined\fi
\newcount\msipdfoutput
\ifx\pdfoutput\undefined\else
\ifcase\pdfoutput\else
\ifx\paperwidth\undefined\else
\ifdim\paperheight=0pt\relax\else\pdfpageheight\paperheight\fi
\ifdim\paperwidth=0pt\relax\else\pdfpagewidth\paperwidth\fi
\fi\fi\fi
\begin{document}
\title[The $p$-Adic Dirac Equation]{$\wp$-Adic Quantum Mechanics, the Dirac Equation, and the violation of
Einstein causality}
\author[Z\'{u}\~{n}iga-Galindo]{W. A. Z\'{u}\~{n}iga-Galindo}
\address{University of Texas Rio Grande Valley\\
School of Mathematical \& Statistical Sciences\\
One West University Blvd\\
Brownsville, TX 78520, United States}
\email{wilson.zunigagalindo@utrgv.edu}
\thanks{The author was partially supported by the Debnath Endowed \ Professorship}
\subjclass{Primary: 81Q35, 81Q65. Secondary: 26E30}

\begin{abstract}
This article studies the breaking of the Lorentz symmetry at the Planck length
in quantum mechanics. We use three-dimensional $\wp$-adic vectors as position
variables, while the time remains a real number. In this setting, the Planck
length is $1/\wp$, where $\wp$ is a prime number, and the Lorentz symmetry is
naturally broken. In the framework of the Dirac-von Neumann formalism for
quantum mechanics, we introduce a new $\wp$-adic Dirac equation that predicts
the existence of particles and antiparticles and charge conjugation like the
standard one. The discreteness of the $\wp$-adic space imposes substantial
restrictions on the solutions of the new equation. This equation admits
localized solutions, which is impossible in the standard case. We show that an
isolated quantum system whose evolution is controlled by the $\wp$-adic Dirac
equation does not satisfy the Einstein causality, which means that the speed
of light is not\ the upper limit for the speed at which conventional matter or
energy can travel through space. The new $\wp$-adic Dirac equation is not
intended to replace the standard one; it should be understood as a new version
(or a limit) of the classical equation at the Planck length scale.

\end{abstract}
\keywords{$p$-adic numbers, quantum mechanics, Dirac equation, breaking of the Lorentz
symmetry, Einstein causality.}
\maketitle

\section{Introduction}

This article revolves around the breaking of the Lorentz symmetry, models of
quantum mechanics (QM) at the Planck scale, and the Volovich conjecture on the
$\wp$-adic nature of space-time at the Planck scale. The Lorentz symmetry is
one of the most essential symmetries of the quantum field theory. While the
validity of this symmetry continues to be verified with a high degree of
precision \cite{Kostelecky-Russell}, in the last thirty-five years, the
experimental and theoretical studies of the Lorentz breaking symmetry have
been an area of intense research, see, e.g., the reviews \cite{Mariz et
al}-\cite{Amelino-Camelia} and the references therein. The quantum-gravity
perspective is considered in \cite{Doplicher et al}-\cite{Garay}.\ When both
quantum mechanics and general relativity are considered, there should be a
Planck scale limitation to the localization of a space-time point. This
naturally leads one to consider discretized space-times.

In the Dirac-von Neumann formulation of QM, the states of a quantum system are
vectors of an abstract complex Hilbert space $\mathcal{H}$, and the
observables correspond to linear self-adjoint operators in $\mathcal{H}$,
\cite{Dirac}-\cite{Berezin et al}. A particular choice of space $\mathcal{H}$
goes beyond the mathematical formulation and belongs to the domain of physical
practice and intuition, \cite[Chap. 1, Sect. 1.1]{Berezin et al}. Let
$\Psi_{0}\in\mathcal{H}$ be the state at time $t=0$ of a specific quantum
system. Then, at any time $t$, the system is represented by the vector
$\Psi\left(  t\right)  =\boldsymbol{e}^{-it\boldsymbol{H}}\Psi_{0}$, $t\geq0,$
where $\boldsymbol{H}$\ is the observable associated with the total energy of
the system. It is crucial to mention that this description of the evolution of
quantum states requires real-time ($t\in\mathbb{R}$). Nowadays, we do not have
a convenient notion of unitary semigroup with $\wp$-adic time.

By $\wp$-adic quantum mechanics, where $\wp$ is a fixed prime number, we mean
QM with $\mathcal{H}=L^{2}(\mathbb{Q}_{\wp}^{3})$, where $\mathbb{Q}%
_{\mathbb{\wp}}$ is the field of $\wp$-adic numbers. Since the time and the
position are not interchangeable ($t\in\mathbb{R}$, $\boldsymbol{x}%
\in\mathbb{Q}_{\wp}^{3}$), such theory is not Lorentz-invariant. The geometry
of $\mathbb{Q}_{\wp}^{3}$ radically differs from that of $\mathbb{R}^{3}$. The
$\wp$-adic space $\mathbb{Q}_{\wp}^{3}$ is a completely disconnected
topological space, while $\mathbb{R}^{3}$ is an arcwise connected topological
space. Intuitively, $\mathbb{Q}_{\wp}^{3}$ has discrete geometry, while
$\mathbb{R}^{3}$ has a continuous one; see \cite{Zuniga-double-slit} for a
further discussion. In addition, there is no algebraic and topological
isomorphism between $\mathbb{Q}_{\wp}^{m}$ and $\mathbb{R}^{m}$, because as
topological fields $\mathbb{Q}_{\wp}$ is not isomorphic to $\mathbb{R}$,
\cite[Chapter I, Sections 3, 4]{We}. We propose the $\wp$-adic QM as a model
of QM at the Planck length.

The $\wp$-adic QM is a model of the standard QM in the space $\mathbb{Q}_{\wp
}^{3}$. This space is invariant under the transformations of the form
$\boldsymbol{x}\rightarrow\boldsymbol{a}+\boldsymbol{A}\boldsymbol{x}$, where
$\boldsymbol{A}\in GL_{3}\left(  \mathbb{Q}_{\wp}\right)  $. This group is the
substitute of the Poincar\'{e} group in the $\wp$-adic framework. In
particular, \ this group contains scale transformations of the type
$\boldsymbol{x}\rightarrow\boldsymbol{a}+\wp^{L}\boldsymbol{x}$, where
$\boldsymbol{a}\in\mathbb{Q}_{\wp}^{3}$, and $L\in\mathbb{Z}$. Given two
different points $\boldsymbol{x},\boldsymbol{y}\in\mathbb{Q}_{\wp}^{3}$, after
a suitable scale transformation, one may assume that $\boldsymbol{x=a}%
+\wp\boldsymbol{b}\neq\boldsymbol{0}$, with $\boldsymbol{a},\boldsymbol{b}$
$\in\left\{  0,1,\ldots,\wp-1\right\}  ^{3}$, $\boldsymbol{y=0}$. Then, the
distance between these points is $\left\Vert \boldsymbol{a}+\wp^{L}%
\boldsymbol{b}\right\Vert _{\wp}$, and this quantity can take two values, $1$
or $\wp^{-1}$, which means that the Planck length is exactly $\wp^{-1}$, see
also \cite[Section 5]{Zuniga-double-slit}. It is relevant to reemphasize that
the existence of a Planck length is invalid if we replace $\mathbb{Q}_{\wp
}^{3}$ with $\mathbb{R}^{3}$, and also that the $\wp$-adic QM is not a
replacement of the standard QM, but a model of the standard one at the Planck length.

In the last forty years $\wp$-adic quantum mechanics has been studied
intensively; see, e.g., \cite{Beltrameti et al}-\cite{Aniello et al}, among
many available references. There are several different types of $\wp$-adic QM.
In particular, the one in which the time is $\wp$-adic, e.g., \cite{V-V-QM3},
radically differs from the one considered here. If the time is $\wp$-adic, we
cannot use the classical theory of semigroups; then, a new theory for the
computation of quantum expectation values is required. Using real-time allows
us to do calculations using the standard axioms of QM, with the same physical
interpretations but in a discrete space.

In the 1930s Bronstein showed that the general relativity and quantum
mechanics imply that the uncertainty $\Delta x$ of any length measurement
satisfies
\begin{equation}
\Delta x\geq L_{\text{Planck}}, \label{Inequality}%
\end{equation}
where $L_{\text{Planck}}$ is the Planck length, \cite{Bronstein}. This
inequality establishes an absolute limitation on length measurements, so the
Planck length is the smallest possible distance that can, in principle, be
measured. Below the Planck scale there are no intervals just points. The
choice of $\mathbb{R}$ as a model the unidimensional space is not compatible
with inequality (\ref{Inequality}) because, $\mathbb{R}$ contains intervals
with arbitrarily small length. On the other hand, there are no intervals in
$\mathbb{Q}_{\wp}$, i.e., the non-trivial connected subsets are points. So
$\mathbb{Q}_{\wp}$ is the prototype of a discrete space with a very rich
mathematical structure. This idea is a reformulation of the Volovich
conjecture on the $\wp$-adic nature of the space at the Planck scale,
\cite{Volovich}. It is relevant to mention that another interpretation of
Bronstein's inequality drives to quantum gravity, \cite{Rovelli et al}.

Since 1925, all the parents of the QM have been aware of the need to abandon
the classical notion of continuous space-time in studying phenomena at the
quantum scale. The reader may consult \cite{Capellmann} for an in-depth
historical review. For instance, according to \cite{Capellmann}, Einstein and
Born believed that the traditional concept of space-time of macroscopic
physics cannot simply be transferred to quantum physics. But all the efforts
(especially from Born and Jordan) to construct physical theories with discrete
space-times failed.

The theoretical study of quantum models admitting Lorentz symmetry breaking is
relevant, \cite{Mariz et al}-\cite{Amelino-Camelia}. Here, we review some
ideas presented in \cite{Garay} directly related to this discussion. Based on
the inequality $E^{-1}\Delta E\gtrsim1$, where $\Delta E$ is the uncertainty
in energy measurement, where $\hbar=c=1$ so the Planck length is $\sqrt{G}$,
in \cite{Garay} is argued that the uncertainty in the energy of a particle is
more significant than its rest mass and this makes the concept of particle
unclear. The Planck length imposes a resolution limit for relativistic quantum
mechanics: localizing a particle with better accuracy than its Compton
wavelength is impossible. These features are reflected in the difficulty of
defining a position operator (in the sense of Newton and Wigner). This
operator, whose eigenvalues give the positions of a certain particle, has the
following property. If a particle is localized in a certain region at a
certain time, then at any arbitrarily close instant of time, there is a
non-zero probability of finding it anywhere; therefore, the particle would
travel faster than light. Finally, the author argues, without giving
additional details, that this paradox disappears when, instead of dealing with
a single particle, one allows the possibility of particle creation and annihilation.

The simplest quantum model of the phenomena\ just described is the $\wp$-adic
Dirac equation. This article introduces a such equation which shares many
properties with the standard one. In particular, the new equation also
predicts the existence of pairs of particles and antiparticles and a charge
conjugation symmetry. The $\wp$-adic Dirac spinors depend on the standard
Pauli-Dirac matrices. The new equation is a version of the standard Dirac
equation at the Planck scale, where the breaking of the Lorentz symmetry
naturally occurs. In this framework, the Einstein causality is not valid.

The derivation of the $\wp$-adic Dirac equation is based on the fact that the
plane wave solutions of the standard Dirac equation have natural analogs when
one considers the position and momenta as elements of a metric, locally
compact topological group; the construction of these analogs does not require
Lorentz invariance, just a version of the relativistic energy formula, with
$c=1$. This last normalization, or a similar one, is essential because, in the
new theory, the speed of light is not the upper bound for the speed at which
conventional matter or energy can travel through space. We warn the reader
that our relativistic energy formula with $c=1$ is not a substitute for the
classical one. It provides the proper geometric restriction for the existence
of $\wp$-adic plane waves. We do not discuss the Planck units in the
space-time $\mathbb{R}\times\mathbb{Q}_{\wp}^{N}$; instead, we use `natural
units' to simplify the discussion of our model.

Our $\wp$-adic Dirac equation has the form%
\begin{equation}
i\frac{\partial}{\partial t}\Psi\left(  t,\boldsymbol{x}\right)  =\left(
\boldsymbol{\alpha}\cdot\nabla_{\wp}+\beta m\right)  \Psi\left(
t,\boldsymbol{x}\right)  ,\text{ }t\in\mathbb{R}\text{, }\boldsymbol{x}%
\in\mathbb{Q}_{\wp}^{3}, \label{Dirac_1}%
\end{equation}
where $\boldsymbol{\alpha}$, $\cdot$, $m$ has the standard meaning,
\[
\Psi^{T}\left(  t,\boldsymbol{x}\right)  =\left[
\begin{array}
[c]{cccc}%
\Psi_{1}\left(  t,\boldsymbol{x}\right)  & \Psi_{2}\left(  t,\boldsymbol{x}%
\right)  & \Psi_{3}\left(  t,\boldsymbol{x}\right)  & \Psi_{4}\left(
t,\boldsymbol{x}\right)
\end{array}
\right]  \in\mathbb{C}^{4},
\]
$\nabla_{\wp}^{T}=\left[
\begin{array}
[c]{ccc}%
\boldsymbol{D}_{x_{1}} & \boldsymbol{D}_{x_{2}} & \boldsymbol{D}_{x_{3}}%
\end{array}
\right]  $,\ where $\boldsymbol{D}_{x_{k}}$ denotes the Taibleson-Vladimirov
operator. We start with an ansatz that describes the $\wp$-adic counterparts
of the classical plane waves in terms of the classical Pauli-Dirac matrices
and a version of the relativistic energy formula, with $c=1$. Using these
particular solutions, we derive a new $\wp$-adic Dirac equation. It turns out
that the $\wp$-adic Dirac equation shares many properties with the usual one.
Indeed, we use several results from \cite[Chapter 1]{Thaller}; our notation
follows closely the one used in this reference.

The geometry of space $\mathbb{Q}_{\wp}^{3}$ imposes substantial restrictions
on the solutions of (\ref{Dirac_1}). The $\wp$-adic Dirac equation admits
space-localized planes waves $\Psi_{\boldsymbol{rnj}}\left(  t,\boldsymbol{x}%
\right)  $ for any time $t\geq0$, which is, $\mathrm{supp}$ $\Psi
_{\boldsymbol{rnj}}\left(  t,\boldsymbol{\cdot}\right)  $ is contained in a
compact subset of $\mathbb{Q}_{\wp}^{3}$; see Theorem \ref{Theorem3}. This
phenomenon does not occur in the standard case; see, e.g., \cite[Section 1.8,
Corollary 1.7]{Thaller}. On the other hand, we compute the transition
probability from a localized state at time $t=0$ to another localized state at
$t>0$, assuming that the space supports of the states are arbitrarily far
away. It turns out that this transition probability is greater than zero for
any time $t\in\left(  0,\epsilon\right)  $, for arbitrarily small $\epsilon$;
see Theorem \ref{Theorem5}. Since this probability is nonzero for some
arbitrarily small $t$, the system has a nonzero probability of getting between
the mentioned localized states arbitrarily shortly, thereby propagating with
superluminal speed.

The concept of discrete space-time is not uniform in the literature, and thus,
many different QM over discrete spaces occur; see, e.g., \cite{Stovicek}-
\cite{Gudder}. For instance, in \cite[Section 3]{Gudder}, the author considers
discrete the space-time\ of type $\mathbb{Z}^{m}\subset\mathbb{R}^{m}$, i.e.,
the space-time is a lattice of the standard Euclidean space. This approach is
not convenient here because the discreteness of $\mathbb{Z}^{m}$ is a relative
property in $\mathbb{R}^{m}$; this choice does not change the Poincar\'{e}
group of $\mathbb{R}^{m}$, and the discrete space is not invariant under all
these symmetries. Finally, the discussion about the Planck length requires a
mathematical framework where the notion of length measurements using rational
numbers may be formulated. For all these reasons, the approach considered in
\cite{Stovicek}- \cite{Gudder} is not useful here.

The article is organized as follows. Section \ref{Section_1} gives a quick
review of the results of $\wp$-the analysis required here. In Section
\ref{Section_2}, we derive the $\wp$-adic Dirac equation. The\ Section
\ref{Section_3} is dedicated to studying the spectrum of the $\wp$-adic Dirac
operator. The charge conjugation is studied in Section \ref{Section_4}, while
Section \ref{Section_5} is dedicated to showing the existence of localized
particles and antiparticles. The semigroup attached to the\ free Dirac
operator is studied in Section \ref{Section_6}. The position and momenta
operators are considered in Section \ref{Section_7}, and Section
\ref{Section_8} is dedicated to the violation of Einstein causality. In the
last section, we give some final comments and conclusions.

\section{\label{Section_1}Basic facts on $\wp$-adic analysis}

In this section we fix the notation and collect some basic results on $\wp
$-adic analysis that we will use through the article. For a detailed
exposition on $\wp$-adic analysis the reader may consult \cite{V-V-Z},
\cite{Alberio et al}-\cite{Taibleson}.

\subsection{The field of $\wp$-adic numbers}

Along this article $\wp$ denotes a prime number. The field of $\wp-$adic
numbers $\mathbb{Q}_{\wp}$ is defined as the completion of the field of
rational numbers $\mathbb{Q}$ with respect to the $\wp-$adic norm
$|\cdot|_{\wp}$, which is defined as
\[
|x|_{\wp}=%
\begin{cases}
0 & \text{if }x=0\\
\wp^{-\gamma} & \text{if }x=\wp^{\gamma}\dfrac{a}{b},
\end{cases}
\]
where $a$ and $b$ are integers coprime with $\wp$. The integer $\gamma
=ord_{\wp}(x):=ord(x)$, with $ord(0):=+\infty$, is called the\textit{\ }$\wp
-$\textit{adic order of} $x$. We extend the $\wp-$adic norm to $\mathbb{Q}%
_{\wp}^{N}$ by taking%
\[
||\boldsymbol{x}||_{\wp}:=\max_{1\leq i\leq N}|x_{i}|_{\wp},\qquad\text{for
}\boldsymbol{x}=(x_{1},\dots,x_{N})\in\mathbb{Q}_{\wp}^{N}.
\]
By defining $ord(\boldsymbol{x})=\min_{1\leq i\leq N}\{ord(x_{i})\}$, we have
$||\boldsymbol{x}||_{\wp}=\wp^{-ord(\boldsymbol{x})}$.\ The metric space
$\left(  \mathbb{Q}_{\wp}^{N},||\cdot||_{\wp}\right)  $ is a complete
ultrametric space. As a topological space $\mathbb{Q}_{\wp}$\ is homeomorphic
to a Cantor-like subset of the real line; see, e.g., \cite{V-V-Z},
\cite{Alberio et al}.

Any $\wp-$adic number $x\neq0$ has a unique expansion of the form
\begin{equation}
x=\wp^{ord(x)}\sum_{j=0}^{\infty}x_{j}\wp^{j}, \label{p-adic-exapansion}%
\end{equation}
where $x_{j}\in\{0,1,\dots,\wp-1\}$ and $x_{0}\neq0$. In addition, any
$\boldsymbol{x}\in\mathbb{Q}_{\wp}^{N}\smallsetminus\left\{  0\right\}  $ can
be represented uniquely as $\boldsymbol{x}=\wp^{ord(\boldsymbol{x}%
)}\boldsymbol{v}$, where $\left\Vert \boldsymbol{v}\right\Vert _{\wp}=1$.

\subsection{Topology of $\mathbb{Q}_{\wp}^{N}$}

For $r\in\mathbb{Z}$, denote by $B_{r}^{N}(\boldsymbol{a})=\{\boldsymbol{x}%
\in\mathbb{Q}_{\wp}^{N};||\boldsymbol{x}-\boldsymbol{a}||_{\wp}\leq\wp^{r}\}$
the ball of radius $\wp^{r}$ with center at $\boldsymbol{a}=(a_{1},\dots
,a_{N})\in\mathbb{Q}_{\wp}^{N}$, and take $B_{r}^{N}(\boldsymbol{0}%
):=B_{r}^{N}$. Note that $B_{r}^{N}(\boldsymbol{a})=B_{r}(a_{1})\times
\cdots\times B_{r}(a_{N})$, where $B_{r}(a_{i}):=\{x\in\mathbb{Q}_{\wp}%
;|x_{i}-a_{i}|_{\wp}\leq\wp^{r}\}$ is the one-dimensional ball of radius
$\wp^{r}$ with center at $a_{i}\in\mathbb{Q}_{\wp}$. The ball $B_{0}^{N}$
equals the product of $N$ copies of $B_{0}=\mathbb{Z}_{\wp}$, the ring of
$\wp-$adic integers. A polydisc is a set of the form%
\[
B_{r_{1}}(a_{1})\times\cdots\times B_{r_{N}}(a_{N}).
\]
We denote by $S_{r}^{N}(\boldsymbol{a})=\{\boldsymbol{x}\in\mathbb{Q}_{\wp
}^{N};||\boldsymbol{x}-\boldsymbol{a|}|_{\wp}=\wp^{r}\}$ the sphere of radius
$\wp^{r}$ with center at $\boldsymbol{a}=(a_{1},\dots,a_{N})\in\mathbb{Q}%
_{\wp}^{N}$, and take $S_{r}^{N}(\boldsymbol{0}):=S_{r}^{N}$. We notice that
$S_{0}^{1}=\mathbb{Z}_{\wp}^{\times}$ (the group of units of $\mathbb{Z}_{\wp
}$), but $\left(  \mathbb{Z}_{\wp}^{\times}\right)  ^{N}\subsetneq S_{0}^{N}$.
The balls and spheres are both open and closed subsets in $\mathbb{Q}_{\wp
}^{N}$. In addition, two balls in $\mathbb{Q}_{\wp}^{N}$ are either disjoint
or one is contained in the other.

As a topological space $\left(  \mathbb{Q}_{\wp}^{N},||\cdot||_{\wp}\right)  $
is totally disconnected, i.e., the only connected \ subsets of $\mathbb{Q}%
_{\wp}^{N}$ are the empty set and the points. A subset of $\mathbb{Q}_{\wp
}^{N}$ is compact if and only if it is closed and bounded in $\mathbb{Q}_{\wp
}^{N}$; see, e.g., \cite[Section 1.3]{V-V-Z}, or \cite[Section 1.8]{Alberio et
al}. The balls and spheres are compact subsets. Thus $\left(  \mathbb{Q}_{\wp
}^{N},||\cdot||_{\wp}\right)  $ is a locally compact topological space.

\begin{notation}
We will use $\Omega\left(  \wp^{-r}||\boldsymbol{x}-\boldsymbol{a}||_{\wp
}\right)  $ to denote the characteristic function of the ball $B_{r}%
^{N}(\boldsymbol{a})=\boldsymbol{a}+\wp^{-r}\mathbb{Z}_{\wp}^{N}$, where
\[
\mathbb{Z}_{\wp}^{N}=\left\{  \boldsymbol{x}\in\mathbb{Q}_{\wp}^{N};\left\Vert
\boldsymbol{x}\right\Vert _{\wp}\leq1\right\}
\]
is the $N$-dimensional unit ball. For more general sets, we will use the
notation $1_{A}$ for the characteristic function of set $A$.
\end{notation}

\subsection{The Haar measure}

Since $(\mathbb{Q}_{\wp}^{N},+)$ is a locally compact topological group, there
exists a Haar measure $d^{N}\boldsymbol{x}$, which is invariant under
translations, i.e., $d^{N}(\boldsymbol{x}+\boldsymbol{a})=d^{N}\boldsymbol{x}%
$, \cite{Halmos}. If we normalize this measure by the condition $\int
_{\mathbb{Z}_{\wp}^{N}}d^{N}\boldsymbol{x}=1$, then $d^{N}\boldsymbol{x}$ is unique.

\subsection{The Bruhat-Schwartz space}

A complex-valued function $\varphi$ defined on $\mathbb{Q}_{\wp}^{N}$ is
called locally constant if for any $\boldsymbol{x}\in\mathbb{Q}_{\wp}^{N}$
there exist an integer $l(\boldsymbol{x})\in\mathbb{Z}$ such that%
\begin{equation}
\varphi(\boldsymbol{x}+\boldsymbol{x}^{\prime})=\varphi(\boldsymbol{x})\text{
for any }\boldsymbol{x}^{\prime}\in B_{l(\boldsymbol{x})}^{N}.
\label{local_constancy}%
\end{equation}
A function $\varphi:\mathbb{Q}_{\wp}^{N}\rightarrow\mathbb{C}$ is called a
Bruhat-Schwartz function (or a test function) if it is locally constant with
compact support. Any test function can be represented as a linear combination,
with complex coefficients, of characteristic functions of balls. The
$\mathbb{C}$-vector space of Bruhat-Schwartz functions is denoted by
$\mathcal{D}(\mathbb{Q}_{\wp}^{N})$. For $\varphi\in\mathcal{D}(\mathbb{Q}%
_{\wp}^{N})$, the largest number $l=l(\varphi)$ satisfying
(\ref{local_constancy}) is called the exponent of local constancy (or the
parameter of constancy) of $\varphi$.

\subsection{$L^{\rho}$ spaces}

Given $\rho\in\lbrack1,\infty)$, we denote by$L^{\rho}\left(
\mathbb{Q}
_{\wp}^{N}\right)  :=L^{\rho}\left(
\mathbb{Q}
_{\wp}^{N},d^{N}\boldsymbol{x}\right)  ,$ the $\mathbb{C}-$vector space of all
the complex valued functions $g$ satisfying
\[
\left\Vert g\right\Vert _{\rho}=\left(  \text{ }%
{\displaystyle\int\limits_{\mathbb{Q}_{\wp}^{N}}}
\left\vert g\left(  \boldsymbol{x}\right)  \right\vert ^{\rho}d^{N}%
\boldsymbol{x}\right)  ^{\frac{1}{\rho}}<\infty,
\]
where $d^{N}\boldsymbol{x}$ is the normalized Haar measure on $\left(
\mathbb{Q}_{\wp}^{N},+\right)  $.

If $U$ is an open subset of $\mathbb{Q}_{\wp}^{N}$, $\mathcal{D}(U)$ denotes
the $\mathbb{C}$-vector space of test functions with supports contained in
$U$, then $\mathcal{D}(U)$ is dense in
\[
L^{\rho}\left(  U\right)  =\left\{  \varphi:U\rightarrow\mathbb{C};\left\Vert
\varphi\right\Vert _{\rho}=\left\{
{\displaystyle\int\limits_{U}}
\left\vert \varphi\left(  \boldsymbol{x}\right)  \right\vert ^{\rho}%
d^{N}\boldsymbol{x}\right\}  ^{\frac{1}{\rho}}<\infty\right\}  ,
\]
for $1\leq\rho<\infty$; see, e.g., \cite[Section 4.3]{Alberio et al}.

\subsection{The Fourier transform}

By using expansion (\ref{p-adic-exapansion}), we define the fractional part
of\textit{ }$x\in\mathbb{Q}_{p}$, denoted $\{x\}_{p}$, as the rational number
\[
\left\{  x\right\}  _{\wp}=\left\{
\begin{array}
[c]{lll}%
0 & \text{if} & x=0\text{ or }ord(x)\geq0\\
&  & \\
\wp^{ord(x)}\sum_{j=0}^{-ord(x)-1}x_{j}\wp^{j} & \text{if} & ord(x)<0.
\end{array}
\right.
\]
We now set $\chi_{\wp}(y):=\exp(2\pi i\{y\}_{\wp})$ for $y\in\mathbb{Q}_{\wp}%
$. The map $\chi_{\wp}(\cdot)$ is an additive character on $\mathbb{Q}_{\wp}$,
i.e., a continuous map from $\left(  \mathbb{Q}_{\wp},+\right)  $ into $S$
(the unit circle considered as multiplicative group) satisfying $\chi_{\wp
}(x_{0}+x_{1})=\chi_{\wp}(x_{0})\chi_{p}(x_{1})$, $x_{0},x_{1}\in
\mathbb{Q}_{\wp}$. \ The additive characters of $\mathbb{Q}_{\wp}$ form an
Abelian group which is isomorphic to $\left(  \mathbb{Q}_{\wp},+\right)  $.
The isomorphism is given by $\xi\rightarrow\chi_{\wp}(\xi x)$; see, e.g.,
\cite[Section 2.3]{Alberio et al}.

Set $\boldsymbol{p}\cdot\boldsymbol{x}:=\sum_{j=1}^{N}p_{j}x_{j}$, for
$\boldsymbol{p}=(p_{1},\dots,p_{N})$, $\boldsymbol{x}=(x_{1},\dots
,x_{N})\allowbreak\in\mathbb{Q}_{\wp}^{N}$. The Fourier transform of
$\varphi\in\mathcal{D}(\mathbb{Q}_{\wp}^{N})$ is defined as
\[
\mathcal{F}\varphi(\xi)=%
{\displaystyle\int\limits_{\mathbb{Q}_{\wp}^{N}}}
\chi_{\wp}(\boldsymbol{p}\cdot\boldsymbol{x})\varphi(\boldsymbol{x}%
)d^{N}\boldsymbol{x}\quad\text{for }\boldsymbol{p}\in\mathbb{Q}_{\wp}^{N},
\]
where $d^{N}\boldsymbol{x}$ is the normalized Haar measure on $\mathbb{Q}%
_{\wp}^{N}$. The Fourier transform is a linear isomorphism from $\mathcal{D}%
(\mathbb{Q}_{\wp}^{N})$ onto itself satisfying
\begin{equation}
(\mathcal{F}(\mathcal{F}\varphi))(\boldsymbol{x})=\varphi(-\boldsymbol{x});
\label{Eq_FFT}%
\end{equation}
see, e.g., \cite[Section 4.8]{Alberio et al}. We also use the notation
$\mathcal{F}_{\boldsymbol{x}\rightarrow\boldsymbol{p}}\varphi$ and
$\widehat{\varphi}$\ for the Fourier transform of $\varphi$.

The Fourier transform extends to $L^{2}$. If $f\in L^{2}\left(  \mathbb{Q}%
_{\wp}^{N}\right)  $, its Fourier transform is defined as
\[
(\mathcal{F}f)(\boldsymbol{p})=\lim_{k\rightarrow\infty}%
{\displaystyle\int\limits_{||x||_{\wp}\leq\wp^{k}}}
\chi_{\wp}(\boldsymbol{p}\cdot\boldsymbol{x})f(x)d^{N}\boldsymbol{x}%
,\quad\text{for }\boldsymbol{p}\in%
\mathbb{Q}
_{\wp}^{N},
\]
where the limit is taken in $L^{2}\left(  \mathbb{Q}_{\wp}^{N}\right)  $. We
recall that the Fourier transform is unitary on $L^{2}\left(  \mathbb{Q}_{\wp
}^{N}\right)  $, i.e. $||f||_{2}=||\mathcal{F}f||_{2}$ for $f\in L^{2}\left(
\mathbb{Q}_{\wp}^{N}\right)  $ and that (\ref{Eq_FFT}) is also valid in
$L^{2}\left(  \mathbb{Q}_{\wp}^{N}\right)  $; see, e.g., \cite[Chapter III,
Section 2]{Taibleson}.

\subsection{Distributions}

The $\mathbb{C}$-vector space $\mathcal{D}^{\prime}\left(  \mathbb{Q}_{\wp
}^{N}\right)  $ of all continuous linear functionals on $\mathcal{D}%
(\mathbb{Q}_{\wp}^{N})$ is called the Bruhat-Schwartz space of distributions.
Every linear functional on $\mathcal{D}(\mathbb{Q}_{\wp}^{N})$ is continuous,
i.e. $\mathcal{D}^{\prime}\left(  \mathbb{Q}_{\wp}^{N}\right)  $\ agrees with
the algebraic dual of $\mathcal{D}(\mathbb{Q}_{\wp}^{N})$; see, e.g.,
\cite[Chapter 1, VI.3, Lemma]{V-V-Z}.

We endow $\mathcal{D}^{\prime}\left(  \mathbb{Q}_{\wp}^{N}\right)  $ with the
weak topology, i.e. a sequence $\left\{  T_{j}\right\}  _{j\in\mathbb{N}}$ in
$\mathcal{D}^{\prime}\left(  \mathbb{Q}_{\wp}^{N}\right)  $ converges to $T$
if $\lim_{j\rightarrow\infty}T_{j}\left(  \varphi\right)  =T\left(
\varphi\right)  $ for any $\varphi\in\mathcal{D}(\mathbb{Q}_{\wp}^{N})$. The
map
\[%
\begin{array}
[c]{lll}%
\mathcal{D}^{\prime}\left(  \mathbb{Q}_{\wp}^{N}\right)  \times\mathcal{D}%
(\mathbb{Q}_{\wp}^{N}) & \rightarrow & \mathbb{C}\\
\left(  T,\varphi\right)  & \rightarrow & T\left(  \varphi\right)
\end{array}
\]
is a bilinear form which is continuous in $T$ and $\varphi$ separately. We
call this map the pairing between $\mathcal{D}^{\prime}\left(  \mathbb{Q}%
_{\wp}^{N}\right)  $ and $\mathcal{D}(\mathbb{Q}_{\wp}^{N})$. From now on we
will use $\left(  T,\varphi\right)  $ instead of $T\left(  \varphi\right)  $.

Every $f$\ in $L_{loc}^{1}$ defines a distribution $f\in\mathcal{D}^{\prime
}\left(  \mathbb{Q}_{\wp}^{N}\right)  $ by the formula
\[
\left(  f,\varphi\right)  =%
{\displaystyle\int\limits_{\mathbb{Q}_{\wp}^{N}}}
f\left(  \boldsymbol{x}\right)  \varphi\left(  \boldsymbol{x}\right)
d^{N}\boldsymbol{x}.
\]

\subsection{The Fourier transform of a distribution}

The Fourier transform $\mathcal{F}\left[  T\right]  $ of a distribution
$T\in\mathcal{D}^{\prime}\left(  \mathbb{Q}_{\wp}^{N}\right)  $ is defined by%
\[
\left(  \mathcal{F}\left[  T\right]  ,\varphi\right)  =\left(  T,\mathcal{F}%
\left[  \varphi\right]  \right)  \text{ for all }\varphi\in\mathcal{D}\left(
\mathbb{Q}_{\wp}^{N}\right)  \text{.}%
\]
The Fourier transform $T\rightarrow\mathcal{F}\left[  T\right]  $ is a linear
and continuous isomorphism from $\mathcal{D}^{\prime}\left(  \mathbb{Q}_{\wp
}^{N}\right)  $\ onto $\mathcal{D}^{\prime}\left(  \mathbb{Q}_{\wp}%
^{N}\right)  $. Furthermore, $T\left(  \boldsymbol{x}\right)  =\mathcal{F}%
\left[  \mathcal{F}\left[  T\right]  \left(  -\boldsymbol{x}\right)  \right]
$.

\subsection{The Taibleson-Vladimirov operator}

We denote by $\boldsymbol{D}_{z}$ the Taibleson-Vladimirov derivative, where
$z\in\mathbb{Q}_{\wp}$, which is defined as%
\[
\left(  \boldsymbol{D}_{z}\varphi\right)  \left(  z\right)  =\frac{1-\wp
}{1-\wp^{-2}}%
{\displaystyle\int\nolimits_{\mathbb{Q}_{\wp}}}
\frac{\varphi\left(  z-y\right)  -\varphi\left(  z\right)  }{\left\vert
y\right\vert _{\wp}^{2}}\text{, for }\varphi\in\mathcal{D}\left(
\mathbb{Q}_{\wp}\right)  \text{.}%
\]
$\boldsymbol{D}_{z}$ is an unbounded operator with a dense domain in
$L^{2}\left(  \mathbb{Q}_{\wp}\right)  $.

We denote by $C(\mathbb{Q}_{\wp},\mathbb{C})$ the $\mathbb{C}$-vector space of
\ continuous $\mathbb{C}$-valued functions defined on $\mathbb{Q}_{\wp}$. We
use the symbol $\Omega\left(  t\right)  $ to denote the characteristic
function of the interval $\left[  0,1\right]  $.

The Taibleson-Vladimirov derivative $\boldsymbol{D}_{x_{i}}$ is a
pseudo-differential operator of the form%
\begin{equation}%
\begin{array}
[c]{cccc}%
\boldsymbol{D}_{x_{i}}: & \mathcal{D}\left(  \mathbb{Q}_{\wp}\right)  &
\rightarrow & C(\mathbb{Q}_{\wp},\mathbb{C})\cap L^{2}\left(  \mathbb{Q}_{\wp
}\right) \\
&  &  & \\
& \varphi & \rightarrow & \left(  \boldsymbol{D}_{x_{i}}\varphi\right)
\left(  x_{i}\right)  =\mathcal{F}_{p_{i}\rightarrow x_{i}}^{-1}\left\{
\left\vert p_{i}\right\vert _{p}\mathcal{F}_{x_{i}\rightarrow p_{i}}%
\varphi\right\}  ;
\end{array}
\label{pseudodifferetial_op}%
\end{equation}
see, e.g., \cite[Chapter 2, Section IX]{V-V-Z}, \cite[Section 2.2]{Kochubei}.

The set of functions $\left\{  \psi_{rnj}\right\}  $ defined as%
\begin{equation}
\psi_{rnj}\left(  x_{i}\right)  =\wp^{\frac{-r}{2}}\chi_{\wp}\left(  \wp
^{-1}j\left(  \wp^{r}x_{i}-n\right)  \right)  \Omega\left(  \left\vert \wp
^{r}x_{i}-n\right\vert _{p}\right)  , \label{Eq_8}%
\end{equation}
where $r\in\mathbb{Z}$, $j\in\left\{  1,\cdots,\wp-1\right\}  $, and $n$ runs
through a fixed set of representatives of $\mathbb{Q}_{\wp}/\mathbb{Z}_{\wp}$,
is an orthonormal basis of $L^{2}(\mathbb{Q}_{\wp})$ consisting of
eigenvectors of operator $\boldsymbol{D}_{x_{i}}$:%
\begin{equation}
\boldsymbol{D}_{x_{i}}\psi_{rnj}\left(  x_{i}\right)  =\wp^{1-r}\text{ }%
\psi_{rnj}\left(  x_{i}\right)  \text{ for any }r\text{, }n\text{, }j\text{;}
\label{Eq_9}%
\end{equation}
see, e.g., \cite[Theorem 3.29]{KKZuniga}, \cite[Theorem 9.4.2]{Alberio et
al}.\ Notice that $\psi_{rnj}\left(  x_{i}\right)  $ \ is supported on the
ball%
\begin{equation}
B_{r}\left(  \wp^{-r}n\right)  =\wp^{-r}n+\wp^{-r}\mathbb{Z}_{\wp}=\left\{
z\in\mathbb{Q}_{\wp};\left\vert z-\wp^{-r}n\right\vert \leq\wp^{r}\right\}  .
\label{support_ball}%
\end{equation}

\section{\label{Section_2}$\wp$-adic Pseudo-differential Equations of Dirac
type}

\subsection{A $\wp$-adic version of the Dirac equation}

We denote the Pauli matrices as%
\[
\sigma_{1}=\left[
\begin{array}
[c]{cc}%
0 & 1\\
1 & 0
\end{array}
\right]  ,\text{ \ }\sigma_{2}=\left[
\begin{array}
[c]{cc}%
0 & -i\\
i & 0
\end{array}
\right]  ,\text{ \ }\sigma_{3}=\left[
\begin{array}
[c]{cc}%
1 & 0\\
0 & -1
\end{array}
\right]  ,\text{ \ }%
\]
where $i=\sqrt{-1}\in\mathbb{C}$ , and the $4\times4$ Dirac matrices
\[
\beta=\left[
\begin{array}
[c]{cc}%
\boldsymbol{1} & \boldsymbol{0}\\
\boldsymbol{0} & -\boldsymbol{1}%
\end{array}
\right]  ,\text{ \ \ }\alpha_{k}=\left[
\begin{array}
[c]{cc}%
\boldsymbol{0} & \sigma_{k}\\
\sigma_{k} & \boldsymbol{0}%
\end{array}
\right]  \text{, for }k=1,2,3\text{,}%
\]
where $\boldsymbol{1}$ denotes the $2\times2$ identity matrix, and
$\boldsymbol{0}$ denotes the $2\times2$ zero matrix.

We set%
\[
\nabla_{\wp}:=\left[
\begin{array}
[c]{c}%
\boldsymbol{D}_{x_{1}}\\
\boldsymbol{D}_{x_{2}}\\
\boldsymbol{D}_{x_{3}}%
\end{array}
\right]  ,
\]
and%
\[
\boldsymbol{\alpha}\cdot\nabla_{\wp}:=%
{\displaystyle\sum\limits_{k=1}^{3}}
\alpha_{k}\boldsymbol{D}_{x_{k}}\text{, and }\boldsymbol{\sigma}\cdot
\nabla_{\wp}:=%
{\displaystyle\sum\limits_{k=1}^{3}}
\sigma_{k}\boldsymbol{D}_{x_{k}}\text{.}%
\]
We suppress Einstein's convention because we need just a version of the
`relativistic energy formula.'

We define the free Hamiltonian as the operator%
\begin{equation}
\boldsymbol{H}_{0}:=\boldsymbol{\alpha}\cdot\nabla_{\wp}+\beta m=\left[
\begin{array}
[c]{cc}%
m\boldsymbol{1} & \boldsymbol{\sigma}\cdot\nabla_{\wp}\\
\boldsymbol{\sigma}\cdot\nabla_{\wp} & -m\boldsymbol{1}%
\end{array}
\right]  . \label{Hamiltonian}%
\end{equation}
We assume that the constant $m\in\mathbb{R}$ (the mass) has a similar meaning
as in the standard QM.

The matrix-valued operator $\boldsymbol{H}_{0}$ acts on\ functions
\[
\phi\left(  \boldsymbol{x}\right)  =\left[
\begin{array}
[c]{c}%
\phi_{1}\left(  \boldsymbol{x}\right) \\
\vdots\\
\phi_{4}\left(  \boldsymbol{x}\right)
\end{array}
\right]  \in\mathcal{D}\left(  \mathbb{Q}_{\wp}\right)
{\textstyle\bigoplus}
\mathcal{D}\left(  \mathbb{Q}_{\wp}\right)
{\textstyle\bigoplus}
\mathcal{D}\left(  \mathbb{Q}_{\wp}\right)
{\textstyle\bigoplus}
\mathcal{D}\left(  \mathbb{Q}_{\wp}\right)  :=\mathcal{D}\left(
\mathbb{Q}_{\wp}\right)
{\textstyle\bigotimes}
\mathbb{C}^{4}.
\]
We denote by $\Psi\left(  t,\boldsymbol{x}\right)  $ a vector-valued
wavefunction, where $t\in\mathbb{R}$, $\boldsymbol{x=}\left(  x_{1}%
,x_{2},x_{3}\right)  \in\mathbb{Q}_{\wp}^{3}$. defined as
\[
\Psi\left(  t,\boldsymbol{x}\right)  =\left[
\begin{array}
[c]{c}%
\Psi_{1}\left(  t,\boldsymbol{x}\right) \\
\vdots\\
\Psi_{4}\left(  t,\boldsymbol{x}\right)
\end{array}
\right]  \in\mathbb{C}^{4}.
\]
Our $\wp$-adic version of the Dirac equation has the form%
\begin{equation}
i\frac{\partial}{\partial t}\Psi\left(  t,\boldsymbol{x}\right)
=\boldsymbol{H}_{0}\Psi\left(  t,\boldsymbol{x}\right)  . \label{Dirac_Eq_1}%
\end{equation}
This version of the Dirac equation can be derived using the original Dirac's
argument. Starting with the relativistic energy formula (with $c=1$)%
\[
E^{2}=\left(  p_{x}^{2}+p_{y}^{2}+p_{z}^{2}\right)  +m^{2}=:p_{\mathbb{R}}%
^{2}+m^{2}\text{, }%
\]
with $p_{x}$, $p_{y}$, $p_{z}\in\mathbb{R}$, and using the following adhoc
quantization scheme%
\begin{equation}
E\rightarrow i\frac{\partial}{\partial t}\text{ \ \ \ \ \ \ \ and
\ \ \ \ \ \ }p_{\mathbb{R}}\rightarrow\nabla_{\wp}, \label{quantization}%
\end{equation}
one formally obtains%
\[
i\frac{\partial}{\partial t}\Psi\left(  t,\boldsymbol{x}\right)
=\sqrt{\boldsymbol{D}_{x_{1}}^{2}+\boldsymbol{D}_{x_{2}}^{2}+\boldsymbol{D}%
_{x_{3}}^{2}+m^{2}}\text{ }\Psi\left(  t,\boldsymbol{x}\right)  .
\]
Now, the square root is computed using Dirac's original reasoning. These
calculations are embedded in the demonstration of Proposition \ref{Prop_1}.

\begin{remark}
As discussed in the introduction, the Planck length implies that localizing a
particle with better accuracy than its Compton wavelength ($\lambda
_{\text{Compton}}$) is impossible. Then, the notion of particle does not make
sense. Instead of this notion, we introduce the idea of localized particle,
but for the sake of simplicity, we will use the word particle. The
localization property means that the position function is a locally constant
with an exponent of local constancy controlled by $\lambda_{\text{Compton}}$,
which implies that the assertion the position of a particle is
$\boldsymbol{x=}\left(  x_{1},x_{2},x_{3}\right)  \in\mathbb{Q}_{\wp}^{3}$
means that the particle is in a ball of radius controlled by $\lambda
_{\text{Compton}}$. We will use the expressions localized states or localized
waves to mean that they are functions with compact support.
\end{remark}

\subsection{Plane waves}

Given $\boldsymbol{x=}\left(  x_{1},x_{2},x_{3}\right)  $, $\boldsymbol{p=}%
\left(  p_{1},p_{2},p_{3}\right)  \in\mathbb{Q}_{\wp}^{3}$, we set
\[
\left\vert \underline{\boldsymbol{p}}\right\vert _{\wp}:=\left(  \left\vert
p_{1}\right\vert _{\wp},\left\vert p_{2}\right\vert _{\wp},\left\vert
p_{3}\right\vert _{\wp}\right)  \in\mathbb{R}^{3}\text{, and \ }\left\vert
\underline{\boldsymbol{p}}\right\vert _{\wp}^{2}:=\left\vert p_{1}\right\vert
_{\wp}^{2}+\left\vert p_{2}\right\vert _{\wp}^{2}+\left\vert p_{3}\right\vert
_{\wp}^{2}.
\]
We \ recall that $\left\Vert \boldsymbol{p}\right\Vert _{\wp}=\max_{1\leq
i\leq3}\left\{  \left\vert p_{i}\right\vert _{\wp}\right\}  $, and that
$\boldsymbol{p}\cdot\boldsymbol{x}=\sum_{i=1}^{3}p_{i}x_{i}$.

A $\wp$-adic counterpart of the Dirac equation is worthwhile if it predicts
the existence of particles and antiparticles with spin-$\frac{1}{2}$ on a
space-time of the form $\mathbb{R}\times\mathbb{Q}_{\wp}^{3}$. The following
definition describes our anzatz for $\wp$-adic plane waves.

\begin{definition}
By a plane wave, we mean a function of the form%
\begin{equation}
\Psi\left(  t,\boldsymbol{x}\right)  =e^{-iEt}\chi_{\wp}\left(  \boldsymbol{p}%
\cdot\boldsymbol{x}\right)  w\left(  \boldsymbol{p}\right)  , \label{Eq_1}%
\end{equation}
where
\begin{equation}
E^{2}=\left(  \left\vert p_{1}\right\vert _{\wp}^{2}+\left\vert p_{2}%
\right\vert _{\wp}^{2}+\left\vert p_{3}\right\vert _{\wp}^{2}\right)
+m^{2}=\left\vert \underline{\boldsymbol{p}}\right\vert _{\wp}^{2}%
+m^{2},\nonumber
\end{equation}
and%
\begin{equation}
w\left(  \boldsymbol{p}\right)  =\left[
\begin{array}
[c]{c}%
w_{1}\left(  \boldsymbol{p}\right) \\
\vdots\\
w_{4}\left(  \boldsymbol{p}\right)
\end{array}
\right]  \in\mathbb{C}^{4}. \label{Eq_3}%
\end{equation}
The functions $w\left(  \boldsymbol{p}\right)  $ are `radial,' i.e., $w\left(
\boldsymbol{p}\right)  =w\left(  \left\vert \underline{\boldsymbol{p}%
}\right\vert _{\wp}\right)  $, and they have the form%
\begin{equation}
w_{1}\left(  \boldsymbol{p}\right)  =\left[
\begin{array}
[c]{r}%
\left[
\begin{array}
[c]{c}%
1\\
0
\end{array}
\right] \\
\\
\frac{\boldsymbol{\sigma}\cdot\left\vert \underline{\boldsymbol{p}}\right\vert
_{\wp}}{E+m}\left[
\begin{array}
[c]{c}%
1\\
0
\end{array}
\right]
\end{array}
\right]  \text{,\hspace{2.5cm}}w_{2}\left(  \boldsymbol{p}\right)  =\left[
\begin{array}
[c]{r}%
\left[
\begin{array}
[c]{c}%
0\\
1
\end{array}
\right] \\
\\
\frac{\boldsymbol{\sigma}\cdot\left\vert \underline{\boldsymbol{p}}\right\vert
_{\wp}}{E+m}\left[
\begin{array}
[c]{c}%
0\\
1
\end{array}
\right]
\end{array}
\right]  , \label{Eq_4}%
\end{equation}%
\begin{equation}
w_{3}\left(  \boldsymbol{p}\right)  =\left[
\begin{array}
[c]{r}%
\frac{-\boldsymbol{\sigma}\cdot\left\vert \underline{\boldsymbol{p}%
}\right\vert _{\wp}}{E+m}\left[
\begin{array}
[c]{c}%
0\\
1
\end{array}
\right] \\
\\
\left[
\begin{array}
[c]{c}%
0\\
1
\end{array}
\right]
\end{array}
\right]  \text{,\hspace{2.5cm}}w_{4}\left(  \boldsymbol{p}\right)  =\left[
\begin{array}
[c]{r}%
\frac{-\boldsymbol{\sigma}\cdot\left\vert \underline{\boldsymbol{p}%
}\right\vert _{\wp}}{E+m}\left[
\begin{array}
[c]{c}%
1\\
0
\end{array}
\right] \\
\\
\left[
\begin{array}
[c]{c}%
1\\
0
\end{array}
\right]
\end{array}
\right]  . \label{Eq_5}%
\end{equation}

\end{definition}

The $\wp$-adic plane waves described above are the natural counterparts of the
standard ones; see, e.g., \cite{Bjorken et al}, \cite{Greiner}. The plane
waves for the ordinary Dirac equation have the form%
\begin{equation}
\Psi\left(  t,\boldsymbol{x}_{\mathbb{R}}\right)  =\left(  e^{-iEt}%
{\displaystyle\prod\nolimits_{k=1}^{3}}
e^{\left(  ip_{\mathbb{R}}^{k}x_{\mathbb{R}}^{k}\right)  }\right)  w\left(
\boldsymbol{p}_{\mathbb{R}}\right)  , \label{Eq_6}%
\end{equation}
where $\boldsymbol{x}_{\mathbb{R}}=\left(  x_{\mathbb{R}}^{1},x_{\mathbb{R}%
}^{2},x_{\mathbb{R}}^{3},\right)  $, $\boldsymbol{p}_{\mathbb{R}}=\left(
p_{\mathbb{R}}^{1},p_{\mathbb{R}}^{2},p_{\mathbb{R}}^{3},\right)
\in\mathbb{R}^{3}$. The term $\exp\left(  -iEt\right)  $ is not affected by
the hypothesis of the discreteness of space; for this reason it appears in
(\ref{Eq_1}). This choice implies that a version relativistic energy for
formula for $E$ should be valid in the $\wp$-adic framework. For the other
terms in (\ref{Eq_6}), we use the correspondence%
\begin{equation}
\exp\left(  2\pi ip_{\mathbb{R}}^{k}x_{\mathbb{R}}^{k}\right)  \rightarrow
\exp\left(  2\pi i\left\{  p_{k}x_{k}\right\}  _{\wp}\right)  , \label{Eq_7}%
\end{equation}
where $p_{k},x_{k}\in\mathbb{Q}_{\wp}$. Since $p_{k}x_{k}$ should be
dimensionless quantity, we require a constant $h_{\wp}=1$ with $\ $%
dimension\ $ML^{2}T^{-1}$, so that $\frac{p_{k}x_{k}}{h_{\wp}}$ be a $\wp
$-adic number. In addition, we need the formula%
\begin{equation}
\boldsymbol{D}_{x_{i}}\chi_{\wp}\left(  p_{i}x_{i}\right)  =\boldsymbol{D}%
_{x_{i}}\exp\left(  2\pi i\left\{  p_{i}x_{i}\right\}  _{\wp}\right)
=\left\vert p_{i}\right\vert _{\wp}\chi_{\wp}\left(  p_{i}x_{i}\right)  ,
\label{Formula_1}%
\end{equation}
see \cite[Chapter 2, Section IX, Example 4]{V-V-Z}. Finally, we need a $\wp
$-adic counterpart for the term $w\left(  \boldsymbol{p}_{\mathbb{R}}\right)
$ in (\ref{Eq_6}). Assuming that the Dirac bispinors are the correct
description of particles/antiparticles with spin-$\frac{1}{2}$, one is
naturally driven to use the correspondence
\[
\frac{\boldsymbol{\sigma}\cdot\boldsymbol{p}_{\mathbb{R}}}{E+m}\rightarrow
\frac{\boldsymbol{\sigma}\cdot\left\vert \underline{\boldsymbol{p}}\right\vert
_{\wp}}{E+m}.
\]

\begin{proposition}
\label{Prop_1}The $\wp$-adic Dirac equation admits plane waves of type
(\ref{Eq_1})-(\ref{Eq_5})\ as solutions.
\end{proposition}

\begin{proof}
The demonstration is just a variation of the classical calculation showing the
existence of plane waves for the Dirac equation. \ By replacing $\Psi\left(
t,\boldsymbol{x}\right)  $, see (\ref{Eq_1}), in (\ref{Dirac_Eq_1}), and using
$\frac{\partial}{\partial t}\Psi\left(  t,\boldsymbol{x}\right)
=-iE\Psi\left(  t,\boldsymbol{x}\right)  $, and formula (\ref{Formula_1}), one
obtains that%
\begin{equation}
Ew\left(  \boldsymbol{p}\right)  =\left(  \boldsymbol{\alpha}\cdot\left\vert
\underline{\boldsymbol{p}}\right\vert _{\wp}+\beta m\right)  w\left(
\boldsymbol{p}\right)  . \label{Eigenvalue_problem_1}%
\end{equation}
Which is a system of linear equations in the variables $w_{1}\left(
\boldsymbol{p}\right)  ,\ldots,w_{4}\left(  \boldsymbol{p}\right)  $ with
coefficients in $\mathbb{C}\left[  \left\vert p_{1}\right\vert _{\wp
},\left\vert p_{2}\right\vert _{\wp},\left\vert p_{3}\right\vert _{\wp
}\right]  $, more precisely,
\[
\left[
\begin{array}
[c]{ccccccc}%
-E+m &  & 0 &  & \left\vert p_{3}\right\vert _{\wp} &  & \left\vert
p_{1}\right\vert _{\wp}-i\left\vert p_{2}\right\vert _{\wp}\\
&  &  &  &  &  & \\
0 &  & -E+m &  & \left\vert p_{1}\right\vert _{\wp}+i\left\vert p_{2}%
\right\vert _{\wp} &  & -\left\vert p_{3}\right\vert _{\wp}\\
&  &  &  &  &  & \\
\left\vert p_{3}\right\vert _{\wp} &  & \left\vert p_{1}\right\vert _{\wp
}-i\left\vert p_{2}\right\vert _{\wp} &  & -E-m &  & 0\\
&  &  &  &  &  & \\
\left\vert p_{1}\right\vert _{\wp}+i\left\vert p_{2}\right\vert _{\wp} &  &
-\left\vert p_{3}\right\vert _{\wp} &  & 0 &  & -E-m
\end{array}
\right]  \left[
\begin{array}
[c]{c}%
w_{1}\left(  \boldsymbol{p}\right) \\
\\
w_{2}\left(  \boldsymbol{p}\right) \\
\\
w_{3}\left(  \boldsymbol{p}\right) \\
\\
w_{4}\left(  \boldsymbol{p}\right)
\end{array}
\right]  =\left[
\begin{array}
[c]{c}%
0\\
\\
0\\
\\
0\\
\\
0
\end{array}
\right]  .
\]
The condition for non-trivial solutions for $w\left(  \boldsymbol{p}\right)  $
is that the determinant of this system vanishes:
\[
\left(  m^{2}+\left\vert \underline{\boldsymbol{p}}\right\vert _{\wp}%
^{2}-E^{2}\right)  ^{2}=0.
\]
The calculation of the determinant is the same as in the classical case. Then,
necessarily
\[
E=\pm\sqrt{m^{2}+\left\vert \underline{\boldsymbol{p}}\right\vert _{\wp}^{2}%
}.
\]
We now consider (\ref{Eigenvalue_problem_1}) as\ an eigenvalue/eigenvector
problem in the ring $\mathbb{C}\left[  \left\vert p_{1}\right\vert _{\wp
},\left\vert p_{2}\right\vert _{\wp},\left\vert p_{3}\right\vert _{\wp
}\right]  $. The solution of this problem follows the classical reasoning and
drives to the announced solutions; see, e.g., \cite{Bjorken et al},
\cite{Greiner}.
\end{proof}

\section{\label{Section_3}The free Dirac operator and its spectrum}

\subsection{\label{Section_Some_Function_Spaces}Some function spaces}

This section uses a notation similar to the one used in \cite[Chapter
1]{Thaller} to compare the standard and the $\wp$-adic Dirac operators
quickly. Furthermore, we use several results and calculations in \cite[Chapter
1]{Thaller}. We set%
\begin{align*}
\mathfrak{H}  &  :=L^{2}\left(  \mathbb{Q}_{\wp}\right)
{\textstyle\bigoplus}
L^{2}\left(  \mathbb{Q}_{\wp}\right)
{\textstyle\bigoplus}
L^{2}\left(  \mathbb{Q}_{\wp}\right)
{\textstyle\bigoplus}
L^{2}\left(  \mathbb{Q}_{\wp}\right) \\
&  =L^{2}\left(  \mathbb{Q}_{\wp}\right)
{\textstyle\bigotimes}
\mathbb{C}^{4}=L^{2}\left(  \mathbb{Q}_{\wp}\right)  ^{4},
\end{align*}
and identify the elements of $\mathfrak{H}$\ with column vectors of the form%
\[
\psi\left(  \boldsymbol{x}\right)  =\left[
\begin{array}
[c]{c}%
\psi_{1}\left(  \boldsymbol{x}\right) \\
\vdots\\
\psi_{4}\left(  \boldsymbol{x}\right)
\end{array}
\right]  \text{.}%
\]
The inner product is given by%
\[
\left(  \psi\left(  \boldsymbol{x}\right)  ,\phi\left(  \boldsymbol{x}\right)
\right)  =%
{\displaystyle\int\limits_{\mathbb{Q}_{\wp}^{3}}}
{\displaystyle\sum\limits_{i=1}^{4}}
\psi_{i}\left(  \boldsymbol{x}\right)  \overline{\phi}_{i}\left(
\boldsymbol{x}\right)  d^{3}\boldsymbol{x},
\]
where the bar denotes the complex conjugate, and the norm is given by%
\[
\left\Vert \psi\right\Vert =\sqrt{%
{\displaystyle\sum\limits_{i=1}^{4}}
\text{ }%
{\displaystyle\int\limits_{\mathbb{Q}_{\wp}^{3}}}
\left\vert \psi_{i}\left(  \boldsymbol{x}\right)  \right\vert ^{2}%
d^{3}\boldsymbol{x}}\text{.}%
\]
Given an integrable function $\psi\in\mathfrak{H}$ its Fourier transform is
defined as%
\[
\left(  \mathcal{F}\psi\right)  \left(  \boldsymbol{p}\right)  =\widehat{\psi
}\left(  \boldsymbol{p}\right)  =\left[
\begin{array}
[c]{c}%
\widehat{\psi}_{1}\left(  \boldsymbol{p}\right) \\
\vdots\\
\widehat{\psi}_{4}\left(  \boldsymbol{p}\right)
\end{array}
\right]  \text{,}%
\]
where%
\[
\left(  \mathcal{F}\psi_{i}\right)  \left(  \boldsymbol{p}\right)
=\widehat{\psi_{i}}\left(  \boldsymbol{p}\right)  =%
{\displaystyle\int\limits_{\mathbb{Q}_{\wp}^{3}}}
\chi_{\wp}\left(  \boldsymbol{p\cdot x}\right)  \psi_{i}\left(  \boldsymbol{x}%
\right)  d^{3}\boldsymbol{x}\text{, for }i=1,2,3,4\text{.}%
\]
The Fourier transform extends to a uniquely defined operator (denote as
$\mathcal{F}$) in the Hilbert space $\mathfrak{H}$.

We now introduce a $\wp$-adic analogue of the first Sobolev space. We define
the scalar product%
\[
\left(  \psi,\phi\right)  _{H^{1}}=%
{\displaystyle\int\limits_{\mathbb{Q}_{\wp}^{3}}}
{\displaystyle\sum\limits_{i=1}^{4}}
\widehat{\psi}_{i}\left(  \boldsymbol{p}\right)  \left(  \max\left(
1,\left\Vert \boldsymbol{p}\right\Vert _{\wp}\right)  \right)  ^{\frac{1}{2}%
}\widehat{\phi}_{i}\left(  \boldsymbol{p}\right)  d^{3}\boldsymbol{p,}%
\]
and the corresponding norm $\left\Vert \psi\right\Vert _{H^{1}}=\sqrt{\left(
\psi,\psi\right)  _{H^{1}}}$. We define%
\[
H^{1}\left(  \mathbb{Q}_{\wp}^{3}\right)  =\overline{\left(  \mathcal{D}%
\left(  \mathbb{Q}_{\wp}\right)
{\textstyle\bigotimes}
\mathbb{C}^{4},\left\Vert \cdot\right\Vert _{H^{1}}\right)  },
\]
where the bar de notes the completion \ of $\mathcal{D}\left(  \mathbb{Q}%
_{\wp}\right)
{\textstyle\bigotimes}
\mathbb{C}^{4}$ with respect to the distance induced by $\left\Vert
\cdot\right\Vert _{H^{1}}$. It also verifies\ that%
\[
H^{1}\left(  \mathbb{Q}_{\wp}^{3}\right)  =\left\{  \psi\in\mathfrak{H}%
;\left\Vert \psi\right\Vert _{H^{1}}<\infty\right\}  .
\]
The results about the space $H^{1}\left(  \mathbb{Q}_{\wp}^{3}\right)  $ are
just variations of well-known results about $\wp$-adic Sobolev spaces; see,
e.g., \cite[Section 10.2]{KKZuniga}, \cite{Zuniga et al}, and the references therein.

The Hamiltonian $\boldsymbol{H}_{0}$ is a matrix operator on $\boldsymbol{D}%
_{x_{1}}$, $\boldsymbol{D}_{x_{2}}$, $\boldsymbol{D}_{x_{3}}$. From now on, we
consider the Taibleson-Vladimirov derivative as a pseudo-differential
operator, see (\ref{pseudodifferetial_op}). With the above notation, it
follows that the mapping%
\[%
\begin{array}
[c]{ccc}%
H^{1}\left(  \mathbb{Q}_{\wp}^{3}\right)  & \rightarrow & \mathfrak{H}\\
&  & \\
\psi & \rightarrow & \boldsymbol{H}_{0}\psi
\end{array}
\]
is a well-defined bounded, linear operator.

The operator $\left(  \boldsymbol{H}_{0},H^{1}\left(  \mathbb{Q}_{\wp}%
^{3}\right)  \right)  $ is self-adjoint, and $\left(  \boldsymbol{H}%
_{0},\mathcal{D}\left(  \mathbb{Q}_{\wp}^{3}\right)
{\textstyle\bigotimes}
\mathbb{C}^{4}\right)  $ is essentially self-adjoint, see \cite[Theorem
3.2]{Kochubei}, \cite[Proposition 7]{Zuniga-LNM-2016}. Then, by Stone's
theorem the initial value problem%
\[
\left\{
\begin{array}
[c]{lll}%
i\frac{\partial}{\partial t}\Psi\left(  t,\boldsymbol{x}\right)
=\boldsymbol{H}_{0}\Psi\left(  t,\boldsymbol{x}\right)  , & t\geq0, &
\boldsymbol{x}\in\mathbb{Q}_{\wp}^{3}\\
&  & \\
\Psi\left(  0,\boldsymbol{x}\right)  =\Psi_{0}\left(  \boldsymbol{x}\right)
, &  &
\end{array}
\right.
\]
has a unique solution given by%
\[
\Psi\left(  t,\boldsymbol{x}\right)  =e^{-it\boldsymbol{H}_{0}}\Psi_{0}\left(
\boldsymbol{x}\right)  \text{.}%
\]

\subsection{The free Dirac operator in the Fourier space}

The results presented in this section are analogs of the results of the
standard Dirac operator. In particular, the calculations used here are the
same as the ones given in \cite[Section 1.4.1]{Thaller}. The Hamiltonian
$\boldsymbol{H}_{0}$ is a matrix pseudo-differential operator on
$\boldsymbol{D}_{x_{1}}$, $\boldsymbol{D}_{x_{2}}$, $\boldsymbol{D}_{x_{3}}$
defined on $H^{1}\left(  \mathbb{Q}_{\wp}^{3}\right)  \subset\mathfrak{H}$.
Any such operator is transformed via $\mathcal{F}$ into a matrix
multiplication operator in $\mathfrak{H}$. In the case of $\boldsymbol{H}_{0}%
$, we have
\begin{equation}
\left(  \boldsymbol{H}_{0}\phi\right)  \left(  \boldsymbol{x}\right)
=\mathcal{F}_{\boldsymbol{p}\rightarrow\boldsymbol{x}}^{-1}\left(  h\left(
\boldsymbol{p}\right)  \mathcal{F}_{\boldsymbol{x}\rightarrow\boldsymbol{p}%
}\phi\right)  ,\text{ for }\phi\in H^{1}\left(  \mathbb{Q}_{\wp}^{3}\right)  ,
\label{pseudodifferetial_op_1}%
\end{equation}
where%
\[
h\left(  \boldsymbol{p}\right)  :=\left[
\begin{array}
[c]{cc}%
m\boldsymbol{1} & \boldsymbol{\sigma}\cdot\left\vert \underline{\boldsymbol{p}%
}\right\vert _{\wp}\\
\boldsymbol{\sigma}\cdot\left\vert \underline{\boldsymbol{p}}\right\vert
_{\wp} & -m\boldsymbol{1}%
\end{array}
\right]  .
\]
The matrix $h\left(  \boldsymbol{p}\right)  =h\left(  \left\vert
\underline{\boldsymbol{p}}\right\vert _{\wp}\right)  $ is a $4\times4$
Hermitian matrix which has the eigenvalues%
\begin{align*}
\lambda_{1}\left(  \left\vert \underline{\boldsymbol{p}}\right\vert _{\wp
}\right)   &  =\lambda_{2}\left(  \left\vert \underline{\boldsymbol{p}%
}\right\vert _{\wp}\right)  =-\lambda_{3}\left(  \left\vert \underline
{\boldsymbol{p}}\right\vert _{\wp}\right)  =-\lambda_{4}\left(  \left\vert
\underline{\boldsymbol{p}}\right\vert _{\wp}\right) \\
&  =:\lambda\left(  \left\vert \underline{\boldsymbol{p}}\right\vert _{\wp
}\right)  =\sqrt{\left(  \left\vert p_{1}\right\vert _{\wp}^{2}+\left\vert
p_{2}\right\vert _{\wp}^{2}+\left\vert p_{3}\right\vert _{\wp}^{2}\right)
+m^{2}}.
\end{align*}
We also use the notation $\lambda\left(  \boldsymbol{p}\right)  =\lambda
\left(  \left\vert \underline{\boldsymbol{p}}\right\vert _{\wp}\right)  $.

The unitary transformation $u(\boldsymbol{p})=u\left(  \left\vert
\underline{\boldsymbol{p}}\right\vert _{\wp}\right)  $, which diagonalizes
$h(\boldsymbol{p})$ is%
\[
u(\boldsymbol{p})=\frac{\left(  m+\lambda\left(  \boldsymbol{p}\right)
\right)  \boldsymbol{1}+\beta\boldsymbol{\alpha\cdot}\left\vert \underline
{\boldsymbol{p}}\right\vert _{\wp}}{\sqrt{2\lambda\left(  \boldsymbol{p}%
\right)  \left(  m+\lambda\left(  \boldsymbol{p}\right)  \right)  }}%
=a_{+}\left(  \boldsymbol{p}\right)  \boldsymbol{1}+a_{-}\left(
\boldsymbol{p}\right)  \beta\frac{\boldsymbol{\alpha\cdot}\left\vert
\underline{\boldsymbol{p}}\right\vert _{\wp}}{\sqrt{\left\vert p_{1}%
\right\vert _{\wp}^{2}+\left\vert p_{2}\right\vert _{\wp}^{2}+\left\vert
p_{3}\right\vert _{\wp}^{2}}},
\]%
\[
u^{-1}(\boldsymbol{p})=a_{+}\left(  \boldsymbol{p}\right)  \boldsymbol{1}%
-a_{-}\left(  \boldsymbol{p}\right)  \beta\frac{\boldsymbol{\alpha\cdot
}\left\vert \underline{\boldsymbol{p}}\right\vert _{\wp}}{\sqrt{\left\vert
p_{1}\right\vert _{\wp}^{2}+\left\vert p_{2}\right\vert _{\wp}^{2}+\left\vert
p_{3}\right\vert _{\wp}^{2}}}%
\]
where $\boldsymbol{1}$ is \ the $4\times4$ matrix identity,
\[
a_{\pm}\left(  \boldsymbol{p}\right)  =\frac{1}{\sqrt{2}}\sqrt{1\pm\frac
{m}{\lambda\left(  \boldsymbol{p}\right)  }},
\]
and the diagonal form of $h(\boldsymbol{p})$ is%
\begin{equation}
u^{-1}(\boldsymbol{p})h\left(  \boldsymbol{p}\right)  u(\boldsymbol{p}%
)=\beta\lambda\left(  \boldsymbol{p}\right)  . \label{diagonalization}%
\end{equation}
By using (\ref{pseudodifferetial_op_1}) and (\ref{diagonalization}), the
unitary transformation
\[
\mathcal{W}:=u\mathcal{F}:\mathfrak{H}\rightarrow\mathfrak{H}%
\]
converts the $\wp$-adic Dirac operator $\boldsymbol{H}_{0}$ into a
multiplication operator by the diagonal matrix $\beta\lambda\left(
\boldsymbol{p}\right)  $,%
\begin{equation}
\boldsymbol{H}_{0}=\mathcal{W}^{-1}\beta\lambda\left(  \boldsymbol{p}\right)
\mathcal{W} \label{diagonalization_2}%
\end{equation}
in $\mathfrak{H}$.

\subsection{The spectrum of $\boldsymbol{H}_{0}$}

In the Hilbert space $\mathcal{W}\mathfrak{H}$ the $\wp$-adic Dirac operator
is diagonal, see (\ref{diagonalization_2}). The upper two components of the
wavefunctions belong to positive energies, while the lower two components
belong to the negative energies.\ Following, Thaller's book \cite[Section
1.4.2]{Thaller}, we introduce \ the subspaces\ of positive energies
\ $\mathfrak{H}_{pos}\subset\mathfrak{H}$ spanned\ by vectors $\psi_{pos}$,
and negative energies $\mathfrak{H}_{neg}\subset\mathfrak{H}$ spanned by
vectors $\psi_{neg}$, where%
\[
\psi_{pos}=\mathcal{W}^{-1}\frac{1}{2}\left(  \boldsymbol{1}+\beta\right)
\mathcal{W}\psi\text{, \ }\psi_{neg}=\mathcal{W}^{-1}\frac{1}{2}\left(
\boldsymbol{1}-\beta\right)  \mathcal{W}\psi\text{, \ }\psi\in\mathfrak{H,}%
\]
where $\boldsymbol{1}$ is the $4\times4$ identity matrix. Since $\left(
\boldsymbol{1}+\beta\right)  \left(  \boldsymbol{1}-\beta\right)
=\boldsymbol{0}$, $\mathfrak{H}_{pos}$ is orthogonal to $\mathfrak{H}_{neg}$,
then%
\begin{equation}
\mathfrak{H=H}_{pos}%
{\textstyle\bigoplus}
\mathfrak{H}_{neg}\text{.} \label{partition_1}%
\end{equation}
Taking
\begin{equation}
\phi_{\pm}:=\frac{1}{2}\left(  \boldsymbol{1}\pm\beta\right)  \mathcal{W}\psi,
\label{partition_2}%
\end{equation}
we have
\[
\left(  \psi_{pos},\boldsymbol{H}_{0}\psi_{pos}\right)  =\left(
\mathcal{W}^{-1}\phi_{+},\mathcal{W}^{-1}\beta\lambda\left(  \boldsymbol{p}%
\right)  \phi_{+}\right)  =\left(  \phi_{+},\beta\lambda\left(  \boldsymbol{p}%
\right)  \phi_{+}\right)  =\left(  \phi_{+},\lambda\left(  \boldsymbol{p}%
\right)  \phi_{+}\right)  >0,
\]
which means that $\mathfrak{H}_{pos}$ is invariant under $\boldsymbol{H}_{0}$.
Similarly, one shows that $\mathfrak{H}_{neg}$ is invariant under
$\boldsymbol{H}_{0}$. The orthogonal projection operators onto the
positive/negative energy subspaces are given by%
\begin{equation}
\mathcal{P}_{\substack{pos\\neg}}=\mathcal{W}^{-1}\frac{1}{2}\left(
\boldsymbol{1}\pm\beta\right)  \mathcal{W}=\frac{1}{2}\left(  \boldsymbol{1}%
\pm\frac{\boldsymbol{H}_{0}}{\left\vert \boldsymbol{H}_{0}\right\vert
}\right)  , \label{Opertor_P}%
\end{equation}
where $\boldsymbol{1}$ is the identity \ operator on $\mathfrak{H}$, and
$\left\vert \boldsymbol{H}_{0}\right\vert $ is the pseudo-differential
operator on $\mathfrak{H}$ with symbol%
\[
\sqrt{\left(  \left\vert p_{1}\right\vert _{\wp}^{2}+\left\vert p_{2}%
\right\vert _{\wp}^{2}+\left\vert p_{3}\right\vert _{\wp}^{2}\right)  +m^{2}%
}\text{ }\boldsymbol{1}\text{. }%
\]
We identify $\left\vert \boldsymbol{H}_{0}\right\vert $ with the operator
$\sqrt{\boldsymbol{H}_{0}^{2}}=\sqrt{\left(  \boldsymbol{D}_{x_{1}}%
^{2}+\boldsymbol{D}_{x_{2}}^{2}+\boldsymbol{D}_{x_{3}}^{2}\right)  +m^{2}}$.
Like in the standard case, we have
\[
\boldsymbol{H}_{0}\psi_{\substack{pos\\neg}}=\pm\left\vert \boldsymbol{H}%
_{0}\right\vert \psi_{\substack{pos\\neg}},
\]
and if we define $sgn$ $\boldsymbol{H}_{0}=\frac{\boldsymbol{H}_{0}%
}{\left\vert \boldsymbol{H}_{0}\right\vert }$, then $\boldsymbol{H}%
_{0}=\left\vert \boldsymbol{H}_{0}\right\vert sgn$ $\boldsymbol{H}_{0}$, which
is polar decomposition \ of $\boldsymbol{H}_{0}$.

Again, following the classical case, we define the Foldy-Wouthuysen
transformation as%
\[
\mathcal{U}_{FW}=\mathcal{F}^{-1}\mathcal{W}.
\]
It transforms the free Dirac operator into the pseudo-differential operator%
\begin{align*}
\mathcal{U}_{FW}\boldsymbol{H}_{0}\mathcal{U}_{FW}^{-1}  &  =\left[
\begin{array}
[c]{cc}%
\sqrt{\left(  \boldsymbol{D}_{x_{1}}^{2}+\boldsymbol{D}_{x_{2}}^{2}%
+\boldsymbol{D}_{x_{3}}^{2}\right)  +m^{2}} & \boldsymbol{0}\\
\boldsymbol{0} & -\sqrt{\left(  \boldsymbol{D}_{x_{1}}^{2}+\boldsymbol{D}%
_{x_{2}}^{2}+\boldsymbol{D}_{x_{3}}^{2}\right)  +m^{2}}%
\end{array}
\right] \\
&  =\beta\left\vert \boldsymbol{H}_{0}\right\vert .
\end{align*}
We interpret this formula as the fact that the free Dirac equation \ is
unitarily equivalent to a pair of (two component) square-root Klein-Gordon equations.

\begin{theorem}
The free Dirac operator is essentially self-adjoint on the dense domain
$\mathcal{D}\left(  \mathbb{Q}_{\wp}\right)
{\textstyle\bigotimes}
\mathbb{C}^{4}$ and self-adjoint in the Sobolev space $H^{1}\left(
\mathbb{Q}_{\wp}^{3}\right)  $ Its spectrum $\sigma\left(  \boldsymbol{H}%
_{0}\right)  $ is the union of the essential range of the functions
$\pm\lambda\left(  \boldsymbol{p}\right)  :\mathbb{Q}_{\wp}^{3}\rightarrow
\mathbb{R}$.
\end{theorem}

\begin{remark}
We denote by $\sigma\left(  \boldsymbol{H}_{0}^{Arch}\right)  $ the spectrum
of the standard free Dirac operator, by \cite[Theorem 1.1]{Thaller},
\[
\sigma\left(  \boldsymbol{H}_{0}^{Arch}\right)  =\left(  -\infty,-m\right]
\cup\left[  m,\infty\right)  .
\]
Then $\sigma\left(  \boldsymbol{H}_{0}\right)  \subset\sigma\left(
\boldsymbol{H}_{0}^{Arch}\right)  $. We use the\ normalization $c=1$ in the
Archimdean case too.
\end{remark}

\begin{proof}
By (\ref{diagonalization_2}) the spectrum of $\boldsymbol{H}_{0}$ equals the
spectrum of the multiplication operator $\beta\lambda\left(  \boldsymbol{p}%
\right)  $, which is essential range of the functions \ $\pm\lambda\left(
\boldsymbol{p}\right)  $; see \cite[Section VII.2]{Reed-Simon I}.
\end{proof}

\section{\label{Section_4}Charge conjugation}

Following the standard case, see \cite[Section 1.4.6]{Thaller}, the $\wp$-adic
Dirac operator for a charge $e\in\mathbb{R}$ in an external electromagnetic
field $\left(  \phi,\boldsymbol{A}\right)  \in\mathbb{R}\times\mathbb{R}^{3}$
is given by%
\[
\boldsymbol{H}(e):=\boldsymbol{\alpha}\cdot\left(  \nabla_{\wp}%
-e\boldsymbol{A}\left(  t,\boldsymbol{x}\right)  \right)  +\beta
m+e\phi\left(  t,\boldsymbol{x}\right)  \boldsymbol{1}\text{.}%
\]
We define the charge conjugation $\mathcal{C}$ as the antiunitary
transformation%
\[
\mathcal{C}\Psi=U_{\mathcal{C}}\overline{\Psi},
\]
where $U_{\mathcal{C}}=-i\beta\alpha_{2}$ is a $4\times4$ unitary matrix.

\begin{lemma}
With the above notation, if $\Psi\left(  t,\boldsymbol{x}\right)  $ is a
solution of
\begin{equation}
i\frac{\partial}{\partial t}\Psi\left(  t,\boldsymbol{x}\right)
=\boldsymbol{H}\left(  e\right)  \Psi\left(  t,\boldsymbol{x}\right)  ,
\label{Eq_15}%
\end{equation}
then%
\[
i\frac{\partial}{\partial t}\mathcal{C}\Psi\left(  t,\boldsymbol{x}\right)
=\boldsymbol{H}\left(  -e\right)  \mathcal{C}\Psi\left(  t,\boldsymbol{x}%
\right)  .
\]
Moreover, $\mathcal{C}^{-1}\boldsymbol{H}\left(  e\right)  \mathcal{C}%
=-\boldsymbol{H}\left(  -e\right)  $.
\end{lemma}

\begin{proof}
Taking the complex conjugate in (\ref{Eq_15}), then multiplying by
$-i\beta\alpha_{2}$, and using that $\overline{\alpha}_{1}=\alpha_{1}$,
$\overline{\alpha}_{2}=-\alpha_{2}$, $\overline{\alpha}_{3}=\alpha_{3}$,
$\alpha_{2}^{2}=\boldsymbol{1}$, $\beta\alpha_{k}=\alpha_{k}\beta$, and
$\alpha_{k}\alpha_{j}=-\alpha_{j}\alpha_{k}$ for $k\neq j$, one gets that%
\begin{align*}
-i\frac{\partial}{\partial t}\mathcal{C}\Psi\left(  t,\boldsymbol{x}\right)
&  =-\boldsymbol{\alpha}\cdot\nabla_{\wp}\mathcal{C}\Psi\left(
t,\boldsymbol{x}\right)  +e\boldsymbol{\alpha}\cdot\boldsymbol{A}%
\mathcal{C}\Psi\left(  t,\boldsymbol{x}\right) \\
&  -\beta m\mathcal{C}\Psi\left(  t,\boldsymbol{x}\right)  -e\phi\left(
t,\boldsymbol{x}\right)  \mathcal{C}\Psi\left(  t,\boldsymbol{x}\right)  .
\end{align*}
The announced formulas follow from this calculation.
\end{proof}

Then negative energy subspace of $\boldsymbol{H}\left(  e\right)  $ is
connected via a symmetry transformation with the positive energy subspace of
the Dirac operator $\boldsymbol{H}\left(  -e\right)  $ for a particle with
opposite charge (antiparticle, positron). For $\mathcal{C}\psi\left(
\boldsymbol{x}\right)  $ \ in the positive energy subspace of $\boldsymbol{H}%
\left(  -e\right)  $, by interpreting $\left\vert \mathcal{C}\psi\left(
\boldsymbol{x}\right)  \right\vert ^{2}$ as \textit{a position probability
density}, the equality%
\[
\left\vert \mathcal{C}\psi\left(  \boldsymbol{x}\right)  \right\vert
^{2}=\left\vert \psi\left(  \boldsymbol{x}\right)  \right\vert ^{2}%
\]
shows that the motion of a negative energy electron state $\psi$ is
indistinguishable from that of a positive energy positron. Then, one obtains
the interpretation:%
\[
\text{a state }\psi\in\mathfrak{H}_{neg}\text{ describes an antiparticle with
positive energy.}%
\]
Here, we do not discuss the problem that the Hilbert space $\mathfrak{H}$
contains states which are superpositions of positive and negative energy
states, \cite[Section 1.4.6]{Thaller}.

\section{\label{Section_5}Localized particles and antiparticles}

By using that $\left\{  \psi_{rnj}\right\}  _{rnj}$\ is an orthonormal basis
of $L^{2}\left(  \mathbb{Q}_{\wp}\right)  $ and some well-known results, we
have
\[
L^{2}\left(  \mathbb{Q}_{\wp},dx_{1}\right)  \otimes L^{2}\left(
\mathbb{Q}_{\wp},dx_{2}\right)  \otimes L^{2}\left(  \mathbb{Q}_{\wp}%
,dx_{3}\right)  =L^{2}\left(  \mathbb{Q}_{\wp}^{3},d^{3}\boldsymbol{x}\right)
,
\]
where $\otimes$\ denotes the \ tensor product of Hilbert spaces. Furthermore,%
\begin{equation}
\psi_{\boldsymbol{rnj}}\left(  \boldsymbol{x}\right)  =%
{\displaystyle\prod\limits_{i=1}^{3}}
\psi_{r_{i}n_{i}j_{i}}\left(  x_{i}\right)  \text{, } \label{Eq_10}%
\end{equation}
where $\boldsymbol{r}=(r_{1},r_{2},r_{3})\in\mathbb{Z}^{3}$, $\boldsymbol{n}%
=(n_{1},n_{2},n_{3})\in\left(  \mathbb{Q}_{\wp}/\mathbb{Z}_{\wp}\right)  ^{3}%
$, \ \ $\boldsymbol{j}=(j_{1},j_{2},j_{3})\in\left\{  1,\cdots,\wp-1\right\}
^{3}$, is an orthonormal basis for $L^{2}\left(  \mathbb{Q}_{\wp}^{3}%
,d^{3}\boldsymbol{x}\right)  $, see, e.g., \cite[Chap. II, Proposition 2,
Theorem II.10-(a)]{Reed-Simon I}.

\begin{definition}
By a space-localized plane wave, we mean a function of the form%
\[
\Psi_{\boldsymbol{rnj}}\left(  t,\boldsymbol{x}\right)  =e^{-\frac{iE}{\hbar
}t}\psi_{\boldsymbol{rnj}}\left(  \boldsymbol{x}\right)  w_{\boldsymbol{r}%
}\left(  \wp^{\left(  1-r_{1}\right)  },\wp^{\left(  1-r_{2}\right)  }%
,\wp^{\left(  1-r_{3}\right)  }\right)  ,
\]
where $\boldsymbol{r}=(r_{1},r_{2},r_{3})\in\mathbb{Z}^{3}$,\ $\boldsymbol{n}%
=(n_{1},n_{2},n_{3})\in\left(  \mathbb{Q}_{\wp}/\mathbb{Z}_{\wp}\right)  ^{3}%
$,\ $\boldsymbol{j}=(j_{1},j_{2},j_{3})\in\left\{  1,\cdots,\wp-1\right\}
^{3}$,
\[
E^{2}=\left(  \wp^{2\left(  1-r_{1}\right)  }+\wp^{2\left(  1-r_{2}\right)
}+\wp^{2\left(  1-r_{3}\right)  }\right)  +m^{2},
\]
and $w_{\boldsymbol{r}}\left(  \wp^{\left(  1-r_{1}\right)  },\wp^{\left(
1-r_{2}\right)  },\wp^{\left(  1-r_{3}\right)  }\right)  =w\left(
\wp^{\left(  1-r_{1}\right)  },\wp^{\left(  1-r_{2}\right)  },\wp^{\left(
1-r_{3}\right)  }\right)  $, with $w\left(  \left\vert p_{1}\right\vert _{\wp
},\left\vert p_{2}\right\vert _{\wp},\left\vert p_{3}\right\vert _{\wp
}\right)  $ defined as (\ref{Eq_3})-(\ref{Eq_5}).
\end{definition}

The term space-localized means that $\Psi_{\boldsymbol{rnj}}\left(
t,\boldsymbol{\cdot}\right)  $ has support on the polydisc%
\begin{equation}
B_{r_{1}}\left(  \wp^{-r_{1}}n_{1}\right)  \times B_{r_{2}}\left(  \wp
^{-r_{2}}n_{2}\right)  \times B_{r_{3}}\left(  \wp^{-r_{3}}n_{3}\right)  .
\label{Eq_ball}%
\end{equation}
Notice that the support of $\Psi_{\boldsymbol{rnj}}\left(  t,\boldsymbol{x}%
\right)  $ is
\[
\left[  0,\infty\right)  \times B_{r_{1}}\left(  \wp^{-r_{1}}n_{1}\right)
\times B_{r_{2}}\left(  \wp^{-r_{2}}n_{2}\right)  \times B_{r_{3}}\left(
\wp^{-r_{3}}n_{3}\right)  .
\]

\begin{theorem}
\label{Theorem3} (i) The localized plane wave $\Psi_{\boldsymbol{rnj}}\left(
t,\boldsymbol{x}\right)  $ is a solution of the $\wp$-adic Dirac equation for
any $\boldsymbol{r}$, $\boldsymbol{n}$, $\boldsymbol{j}$. (ii) Set%
\[
\Psi_{\boldsymbol{rnj}}^{\pm}\left(  t,\boldsymbol{x}\right)  =\mathcal{W}%
^{-1}\frac{1}{2}\left(  \boldsymbol{1}\pm\beta\right)  \mathcal{W}%
\Psi_{\boldsymbol{rnj}}\left(  t,\boldsymbol{x}\right)  .
\]
Then, $\Psi_{\boldsymbol{rnj}}^{+}\left(  t,\boldsymbol{x}\right)  $ is a
particle, resp. $\Psi_{\boldsymbol{rnj}}^{-}\left(  t,\boldsymbol{x}\right)  $
is an antiparticle, space-localized in the polydisc (\ref{Eq_ball}). In
particular,
\[
\Psi_{\boldsymbol{rnj}}^{+}\left(  0,\boldsymbol{x}\right)  \in\mathfrak{H}%
_{pos}\text{, }\Psi_{\boldsymbol{rnj}}^{-}\left(  0,\boldsymbol{x}\right)
\in\mathfrak{H}_{neg}.
\]

\end{theorem}

\begin{proof}
(i) It follows from (\ref{Eq_9}), by using the calculations done in the proof
of Proposition \ref{Prop_1}. (ii) It follows from (\ref{partition_1}%
)-(\ref{partition_2}).
\end{proof}

\begin{remark}
In the standard case, a wavefunction with positive energy cannot be initially
localized in a proper subset of $\mathbb{R}^{3}$. Any wavefunction with
positive energy has to be spread over all space ($\mathbb{R}^{3}$) at all
times. In more precise form, for any state $\psi\in\mathfrak{H}_{pos}$ (or
$\mathfrak{H}_{neg}$), the support of $\psi$ is $\mathbb{R}^{3}$; see \cite[
Corollary 1.7]{Thaller}. The falsity of this result implies the violation of
Einstein causality; see \cite[Section 1.8.2]{Thaller}. Then, Theorem
\ref{Theorem3} provides a strong indication that Einstein causality is not
valid in a discrete space ($\mathbb{Q}_{\wp}^{3}$); this result will be
established later.
\end{remark}

\section{\label{Section_6}The free time evolution}

We define the Lizorkin space of test functions of second kind as%
\[
\mathcal{L}\left(  \mathbb{Q}_{\wp}^{3}\right)  =\left\{  \varphi
\in\mathcal{D}\left(  \mathbb{Q}_{\wp}^{3}\right)  ;%
{\displaystyle\int\limits_{\mathbb{Q}_{\wp}^{3}}}
\varphi\left(  \boldsymbol{x}\right)  d^{3}\boldsymbol{x}=0\right\}  =\left\{
\varphi\in\mathcal{D}\left(  \mathbb{Q}_{\wp}^{3}\right)  ;\widehat{\varphi
}\left(  \boldsymbol{0}\right)  =0\right\}  .
\]
This space in dense in $L^{2}\left(  \mathbb{Q}_{\wp}^{3}\right)  $; see
\cite[Theorem 7.4.3]{Alberio et al}.

\begin{theorem}
\label{Theorem4}With the above notation, for $\psi\in\mathcal{L}\left(
\mathbb{Q}_{\wp}^{3}\right)
{\textstyle\bigotimes}
\mathbb{C}^{4}$,%
\begin{align*}
\left(  e^{-i\boldsymbol{H}_{0}t}\mathcal{P}_{\substack{pos\\neg}}\psi\right)
\left(  \boldsymbol{x}\right)   &  =%
{\displaystyle\int\limits_{\mathbb{Q}_{\wp}^{3}}}
\left\{
{\displaystyle\int\limits_{\mathbb{Q}_{\wp}^{3}}}
\chi_{\wp}\left(  \left(  \boldsymbol{x-y}\right)  \cdot\boldsymbol{p}\right)
e^{\mp\lambda\left(  \boldsymbol{p}\right)  t}\left(  \frac{\pm\lambda\left(
\boldsymbol{p}\right)  +\boldsymbol{h}\left(  \boldsymbol{p}\right)
}{2\lambda\left(  \boldsymbol{p}\right)  }\right)  d^{3}\boldsymbol{p}%
\right\}  \psi\left(  \boldsymbol{y}\right)  d^{3}\boldsymbol{y}\\
&  =%
{\displaystyle\int\limits_{\mathbb{Q}_{\wp}^{3}}}
\left\{  \chi_{\wp}\left(  -\boldsymbol{x}\cdot\boldsymbol{p}\right)
e^{\mp\lambda\left(  \boldsymbol{p}\right)  t}\left(  \frac{\pm\lambda\left(
\boldsymbol{p}\right)  +\boldsymbol{h}\left(  \boldsymbol{p}\right)
}{2\lambda\left(  \boldsymbol{p}\right)  }\right)  \widehat{\psi}\left(
\boldsymbol{p}\right)  d^{3}\boldsymbol{p}\right\}  .
\end{align*}

\end{theorem}

\begin{proof}
Like in the standard case, see proof of Theorem 1.2 in \cite{Thaller}, by
taking $t_{\pm}:=t\mp i\epsilon$, we have%
\[
\lim_{\epsilon\rightarrow0}e^{-i\boldsymbol{H}_{0}t_{\pm}}\phi_{\pm
}=e^{-i\boldsymbol{H}_{0}t}\phi_{\pm}\text{ for }\phi_{\pm}\in\mathfrak{H}%
_{\substack{pos\\neg}}.
\]
Then, by using (\ref{diagonalization}) and (\ref{Opertor_P}),
\begin{equation}
\left(  e^{-i\boldsymbol{H}_{0}t_{\pm}}\mathcal{P}_{\substack{pos\\neg}%
}\psi\right)  \left(  \boldsymbol{x}\right)  =\mathcal{F}_{\boldsymbol{p}%
\rightarrow\boldsymbol{x}}^{-1}\left(  e^{\mp i\lambda\left(  \boldsymbol{p}%
\right)  t_{\pm}}\frac{1}{2}\left(  \boldsymbol{1}\pm\frac{\boldsymbol{h}%
\left(  \boldsymbol{p}\right)  }{\lambda\left(  \boldsymbol{p}\right)
}\right)  \mathcal{F}_{\boldsymbol{x}\rightarrow\boldsymbol{p}}\psi\right)
\left(  \boldsymbol{x}\right)  , \label{Eq_17}%
\end{equation}
where $\mathcal{F}$ is the Fourier transform in $\mathfrak{H}=L^{2}\left(
\mathbb{Q}_{\wp}^{3}\right)
{\textstyle\bigotimes}
\mathbb{C}^{4}$. Since $\psi\in\mathcal{L}\left(  \mathbb{Q}_{\wp}^{3}\right)
%
{\textstyle\bigotimes}
\mathbb{C}^{4}$, and $\mathcal{F}_{\boldsymbol{x}\rightarrow\boldsymbol{p}%
}\psi$ is a test function (i.e., an element of $\mathcal{L}\left(
\mathbb{Q}_{\wp}^{3}\right)
{\textstyle\bigotimes}
\mathbb{C}^{4}$) satisfying $\mathcal{F}_{\boldsymbol{x}\rightarrow
\boldsymbol{p}}\psi\left(  \boldsymbol{0}\right)  =0$ in some ball $B_{l}^{N}$
around the origin, then, the function
\[
\left(  \boldsymbol{1}\pm\frac{\boldsymbol{h}\left(  \boldsymbol{p}\right)
}{\lambda\left(  \boldsymbol{p}\right)  }\right)  \mathcal{F}_{\boldsymbol{x}%
\rightarrow\boldsymbol{p}}\psi\left(  \boldsymbol{p}\right)  =0\text{ for
}\boldsymbol{p}\in B_{l}^{N},
\]
is continuous on the support of $\mathcal{F}_{\boldsymbol{x}\rightarrow
\boldsymbol{p}}\psi$, which is a compact subset. Thus, the function%
\[
e^{\mp i\lambda\left(  \boldsymbol{p}\right)  t_{\pm}}\frac{1}{2}\left(
\boldsymbol{1}\pm\frac{\boldsymbol{h}\left(  \boldsymbol{p}\right)  }%
{\lambda\left(  \boldsymbol{p}\right)  }\right)  \mathcal{F}_{\boldsymbol{x}%
\rightarrow\boldsymbol{p}}\psi
\]
is integrable. Consequently, we can rewrite (\ref{Eq_17}) as follows:%
\begin{align*}
\left(  e^{-i\boldsymbol{H}_{0}t_{\pm}}\mathcal{P}_{\substack{pos\\neg}%
}\psi\right)  \left(  \boldsymbol{x}\right)   &  =%
{\displaystyle\int\limits_{\mathbb{Q}_{\wp}^{3}}}
\left\{
{\displaystyle\int\limits_{\mathbb{Q}_{\wp}^{3}}}
\chi_{\wp}\left(  \left(  \boldsymbol{y}-\boldsymbol{x}\right)  \cdot
\boldsymbol{p}\right)  e^{\mp i\lambda\left(  \boldsymbol{p}\right)  t_{\pm}%
}\frac{1}{2}\left(  \boldsymbol{1}\pm\frac{\boldsymbol{h}\left(
\boldsymbol{p}\right)  }{\lambda\left(  \boldsymbol{p}\right)  }\right)
d^{3}\boldsymbol{p}\right\}  \psi\left(  \boldsymbol{y}\right)  d^{3}%
\boldsymbol{y}\\
&  =%
{\displaystyle\int\limits_{\mathbb{Q}_{\wp}^{3}}}
\chi_{\wp}\left(  -\boldsymbol{x}\cdot\boldsymbol{p}\right)  e^{\mp
i\lambda\left(  \boldsymbol{p}\right)  t_{\pm}}\frac{1}{2}\left(
\boldsymbol{1}\pm\frac{\boldsymbol{h}\left(  \boldsymbol{p}\right)  }%
{\lambda\left(  \boldsymbol{p}\right)  }\right)  \widehat{\psi}\left(
\boldsymbol{p}\right)  d^{3}\boldsymbol{p}.
\end{align*}
In order to compute the limit $\epsilon\rightarrow0$, we first observe that
if
\[
\widehat{\psi}\left(  \boldsymbol{p}\right)  =\left[  \widehat{\psi}%
_{1}\left(  \boldsymbol{p}\right)  ,\ldots,\widehat{\psi}_{4}\left(
\boldsymbol{p}\right)  \right]  ^{T},
\]
then
\begin{align*}
&  \chi_{\wp}\left(  -\boldsymbol{x}\cdot\boldsymbol{p}\right)  e^{\mp
i\lambda\left(  \boldsymbol{p}\right)  t_{\pm}}\frac{1}{2}\left(
\boldsymbol{1}\pm\frac{\boldsymbol{h}\left(  \boldsymbol{p}\right)  }%
{\lambda\left(  \boldsymbol{p}\right)  }\right)  \widehat{\psi}\left(
\boldsymbol{p}\right) \\
&  =\chi_{\wp}\left(  -\boldsymbol{y}\cdot\boldsymbol{p}\right)  e^{\mp
i\lambda\left(  \boldsymbol{p}\right)  t_{\pm}}\left[
\begin{array}
[c]{c}%
\frac{1}{2}\widehat{\psi}_{1}\left(  \boldsymbol{p}\right)  +\frac{1}%
{\lambda\left(  \boldsymbol{p}\right)  }%
{\displaystyle\sum\limits_{j=1}^{4}}
A_{j}^{1}\left(  \boldsymbol{p}\right)  \widehat{\psi}_{j}\left(
\boldsymbol{p}\right) \\
\vdots\\
\frac{1}{2}\widehat{\psi}_{4}\left(  \boldsymbol{p}\right)  +\frac{1}%
{\lambda\left(  \boldsymbol{p}\right)  }%
{\displaystyle\sum\limits_{j=1}^{4}}
A_{j}^{4}\left(  \boldsymbol{p}\right)  \widehat{\psi}_{j}\left(
\boldsymbol{p}\right)
\end{array}
\right]  ,
\end{align*}
where the functions $A_{j}^{k}\left(  \boldsymbol{p}\right)  $ are continuous.
Now, taking%
\[
S:=%
{\displaystyle\bigcup\limits_{j=1}^{4}}
\mathrm{supp}\left(  \widehat{\psi}_{j}\right)  \text{,}%
\]
which is compact subset of $\mathbb{Q}_{\wp}^{3}$, we have%
\begin{align*}
&  \left\vert \chi_{\wp}\left(  -\boldsymbol{y}\cdot\boldsymbol{p}\right)
e^{\mp i\lambda\left(  \boldsymbol{p}\right)  t_{\pm}}\left(  \frac{1}%
{2}\widehat{\psi}_{k}\left(  \boldsymbol{p}\right)  +\frac{1}{\lambda\left(
\boldsymbol{p}\right)  }%
{\displaystyle\sum\limits_{j=1}^{4}}
A_{j}^{k}\left(  \boldsymbol{p}\right)  \widehat{\psi}_{j}\left(
\boldsymbol{p}\right)  \right)  \right\vert \\
&  \leq e^{-\epsilon\left\vert \lambda\left(  \boldsymbol{p}\right)
\right\vert }\left\vert \frac{1}{2}\widehat{\psi}_{k}\left(  \boldsymbol{p}%
\right)  +\frac{1}{\lambda\left(  \boldsymbol{p}\right)  }%
{\displaystyle\sum\limits_{j=1}^{4}}
A_{j}^{k}\left(  \boldsymbol{p}\right)  \widehat{\psi}_{j}\left(
\boldsymbol{p}\right)  \right\vert \leq C%
{\displaystyle\sum\limits_{j=1}^{4}}
\left\vert \widehat{\psi}_{j}\left(  \boldsymbol{p}\right)  \right\vert ,
\end{align*}
where%
\[
C:=\frac{1}{2}+\max_{1\leq j,k\leq4}\left\{  \sup_{\boldsymbol{p}\in
S\smallsetminus B_{l}^{N}}\frac{A_{j}^{k}\left(  \boldsymbol{p}\right)
}{\lambda\left(  \boldsymbol{p}\right)  }\right\}  .
\]
Since $\sum_{j=1}^{4}\left\vert \widehat{\psi}_{j}\left(  \boldsymbol{p}%
\right)  \right\vert $ is an integrable function, by using the dominated
convergence theorem, we conclude that%
\begin{align*}
\lim_{\epsilon\rightarrow0}\left(  e^{-i\boldsymbol{H}_{0}t_{\pm}}%
\mathcal{P}_{\substack{pos\\neg}}\psi\right)  \left(  \boldsymbol{x}\right)
&  =%
{\displaystyle\int\limits_{\mathbb{Q}_{\wp}^{3}}}
\left\{
{\displaystyle\int\limits_{\mathbb{Q}_{\wp}^{3}}}
\chi_{\wp}\left(  \left(  \boldsymbol{y}-\boldsymbol{x}\right)  \cdot
\boldsymbol{p}\right)  e^{\mp i\lambda\left(  \boldsymbol{p}\right)  t}%
\frac{1}{2}\left(  \boldsymbol{1}\pm\frac{\boldsymbol{h}\left(  \boldsymbol{p}%
\right)  }{\lambda\left(  \boldsymbol{p}\right)  }\right)  d^{3}%
\boldsymbol{p}\right\}  \psi\left(  \boldsymbol{y}\right)  d^{3}%
\boldsymbol{y}\\
&  =%
{\displaystyle\int\limits_{\mathbb{Q}_{\wp}^{3}}}
\chi_{\wp}\left(  -\boldsymbol{x}\cdot\boldsymbol{p}\right)  e^{\mp
i\lambda\left(  \boldsymbol{p}\right)  t}\frac{1}{2}\left(  \boldsymbol{1}%
\pm\frac{\boldsymbol{h}\left(  \boldsymbol{p}\right)  }{\lambda\left(
\boldsymbol{p}\right)  }\right)  \widehat{\psi}\left(  \boldsymbol{p}\right)
d^{3}\boldsymbol{p}.
\end{align*}

\end{proof}

\section{\label{Section_7}Position and momentum operators}

Like the standard equation, the $\wp$-adic Dirac equation predicts the
existence of particles/antiparticles of spin-$\frac{1}{2}$; for this reason,
we propose $\boldsymbol{H}_{0}$ as the $\wp$-adic counterpart of the operator
for the energy of a free electron. The definition of the self-adjoint
operators for other observables is a highly non-trivial problem.

The position of a particle corresponds to a point $\boldsymbol{x}=\left(
x_{1},x_{2},x_{3}\right)  \in\mathbb{Q}_{\wp}^{3}$. The Fourier transform
\ sends a function $f(\boldsymbol{x})$ to a function $\widehat{f}%
(\boldsymbol{p})$, $\boldsymbol{p}=\left(  p_{1},p_{2},p_{3}\right)
\in\mathbb{Q}_{\wp}^{3}$. For this reason, we identify \ $\boldsymbol{p}%
$\ with the momentum of the particle. We use the quantization $p_{k}%
\rightarrow\boldsymbol{D}_{x_{k}}$, $k=1,2,3$, where $\boldsymbol{D}_{x_{k}}%
$\ is a pseudo-differential operator with symbol $\left\vert p_{k}\right\vert
_{\wp}$. Here, there is an important difference with standard QM. Now, it is
natural to take the position operator as the multiplication by $\left\vert
x_{k}\right\vert _{\wp}$, more precisely,%
\[
Dom\left(  \left\vert x_{k}\right\vert _{\wp}\right)  =\left\{  \psi
\in\mathfrak{H};%
{\displaystyle\int\limits_{\mathbb{Q}_{\wp}^{3}}}
\text{ }%
{\displaystyle\sum\limits_{j=1}^{4}}
\left\vert x_{k}\right\vert _{\wp}\left\vert \psi_{j}\left(  \boldsymbol{x}%
\right)  \right\vert d^{3}\boldsymbol{x}<\infty\right\}  ,
\]
and%
\[
\left(  \left\vert x_{k}\right\vert _{\wp}\psi\right)  \left(  \boldsymbol{x}%
\right)  =\left[
\begin{array}
[c]{c}%
\left\vert x_{k}\right\vert _{\wp}\psi_{1}\left(  \boldsymbol{x}\right) \\
\vdots\\
\left\vert x_{k}\right\vert _{\wp}\psi_{4}\left(  \boldsymbol{x}\right)
\end{array}
\right]  ,
\]
$k=1,2,3$. The position operator $\left\vert x_{k}\right\vert _{\wp}$,
$k=1,2,3$, is self-adjoint. The spectrum of $\left\vert x_{k}\right\vert
_{\wp}$ is the essential range of the function%
\[
\left\vert x_{k}\right\vert _{\wp}:\mathbb{Q}_{\wp}\rightarrow\mathbb{Q}%
\text{,}%
\]
which is the set $\left\{  \wp^{m};m\in\mathbb{Z}\right\}  $.

\subsection{Spectral projections for the position operator}

For $\lambda\in\mathbb{R}$, we define $m_{\lambda}\in\mathbb{Z}$ as the unique
integer number satisfying%
\[
\wp^{m_{\lambda}}\leq\lambda<\wp^{m_{\lambda}+1}.
\]
We also set
\[
\Omega_{\lambda}^{\cdot}\left(  t\right)  =\left\{
\begin{array}
[c]{cc}%
1 & -\infty<t\leq\lambda\\
0 & \text{otherwise.}%
\end{array}
\right.
\]
Then, $\Omega_{\lambda}^{\cdot}:\mathbb{R\rightarrow R}$ is a Borel
measurable, bounded, function. Now,%
\[
\Omega_{\lambda}^{\cdot}\circ\left\vert x_{k}\right\vert _{\wp}=\left\{
\begin{array}
[c]{cc}%
1 & \left\vert x_{k}\right\vert _{\wp}\leq\wp^{m_{\lambda}}\\
0 & \text{otherwise.}%
\end{array}
\right.  =\Omega\left(  \wp^{-m_{\lambda}}\left\vert x_{k}\right\vert _{\wp
}\right)  ,
\]
which is the characteristic function of the ball $\wp^{-m_{\lambda}}%
\mathbb{Z}_{p}$. The spectral projection of $\left\vert x_{k}\right\vert
_{\wp}$ is the operator%
\[
\boldsymbol{E}(B_{m_{\lambda}}):\psi\left(  \boldsymbol{x}\right)
\rightarrow\Omega\left(  \wp^{-m_{\lambda}}\left\vert x_{k}\right\vert _{\wp
}\right)  \psi\left(  \boldsymbol{x}\right)  .
\]
By the functional calculus,
\[
\boldsymbol{E}(B_{m_{\lambda}})=\boldsymbol{E}(B_{m_{\lambda}})^{2},\text{
\ \ }\boldsymbol{E}(B_{m_{\lambda}})^{\ast}=\boldsymbol{E}(B_{m_{\lambda}}),
\]
which means that $\boldsymbol{E}(B_{m_{\lambda}})$ is an orthogonal projection.

Given a polydisc $B_{m_{\lambda_{1}}}\times B_{m_{\lambda_{2}}}\times
B_{m_{\lambda_{3}}}$, we define the operator%
\[
\boldsymbol{E}(B_{m_{\lambda_{1}}}\times B_{m_{\lambda_{2}}}\times
B_{m_{\lambda_{3}}})=\boldsymbol{E}(B_{m_{\lambda_{1}}})\boldsymbol{E}%
(B_{m_{\lambda_{2}}})\boldsymbol{E}(B_{m_{\lambda_{3}}}).
\]
Then the probability of finding the particle in the region $B_{m_{\lambda_{1}%
}}\times B_{m_{\lambda_{2}}}\times B_{m_{\lambda_{3}}}$ is%
\begin{align*}
\left(  \psi,\boldsymbol{E}(B_{m_{\lambda_{1}}}\times B_{m_{\lambda_{2}}%
}\times B_{m_{\lambda_{3}}})\psi\right)   &  =%
{\displaystyle\int\limits_{B_{m_{\lambda_{1}}}\times B_{m_{\lambda_{2}}}\times
B_{m_{\lambda_{3}}}}}
\text{ }\left\vert \psi\left(  \boldsymbol{x}\right)  \right\vert ^{2}%
d^{3}\boldsymbol{x}\\
&  \boldsymbol{=}%
{\displaystyle\int\limits_{B_{m_{\lambda_{1}}}\times B_{m_{\lambda_{2}}}\times
B_{m_{\lambda_{3}}}}}
\text{ }%
{\displaystyle\sum\limits_{j=1}^{4}}
\left\vert \psi_{j}\left(  \boldsymbol{x}\right)  \right\vert ^{2}%
d^{3}\boldsymbol{x}\text{.}%
\end{align*}
Consequently, $\left\vert \psi\left(  \boldsymbol{x}\right)  \right\vert ^{2}$
can be interpreted as the position probability density, but, this
interpretation holds true only for region of the form $B_{m_{\lambda_{1}}%
}\times B_{m_{\lambda_{2}}}\times B_{m_{\lambda_{3}}}$. We do not know if this
interpretation is valid \ for arbitrary Borel subsets of $\mathbb{Q}_{\wp}%
^{3}$. In the standard case the region $B_{m_{\lambda_{1}}}\times
B_{m_{\lambda_{2}}}\times B_{m_{\lambda_{3}}}$ can be replaced by an arbitrary
Borel subset of $\mathbb{R}^{3}$.

\section{\label{Section_8}The violation of Einstein causality}

In the $\wp$-adic framework, the Dirac equation predicts the existence of
localized particles/antiparticles; see Theorem \ref{Theorem3}. We now consider
a single particle and assume that it has the property \ of being localized in
some region of $\mathbb{Q}_{\wp}^{3}$. To show that this property is an
observable, we construct a self-adjoint operator $\boldsymbol{\Pi}_{B}$ in
$\mathfrak{H}$ which describes the two possibilities of being localized within
$B$ or outside $B$. Thus, $\boldsymbol{\Pi}_{B}$ should have only two
eigenvalues $1$ (within $B$), and 0 (outside of $B$), and consequently, it is
a projector operator.

Let $B\subseteq\mathbb{Q}_{\wp}^{3}$ \ be a Borel subset. We set
\[
L^{2}\left(  B\right)  =\left\{  f:B\rightarrow\mathbb{C};\left\Vert
f\right\Vert _{2,B}:=\sqrt{%
{\displaystyle\int\limits_{B}}
\text{ }\left\vert f\left(  \boldsymbol{x}\right)  \right\vert ^{2}%
d^{3}\boldsymbol{x}}<\infty\right\}  .
\]
By extending any function from $L^{2}\left(  B\right)  $ as zero outside of
$B$, we have a continuous embedding%
\[%
\begin{array}
[c]{cc}%
1_{B}:L^{2}\left(  B\right)  & \rightarrow L^{2}\left(  \mathbb{Q}_{\wp}%
^{3}\right) \\
f & \rightarrow1_{B}f,
\end{array}
\]
where $1_{B}$\ is the characteristic function of $B$. We set
\[
\mathfrak{H}(B):=L^{2}\left(  B\right)
{\textstyle\bigotimes}
\mathbb{C}^{4}=L^{2}\left(  B\right)
{\textstyle\bigoplus}
L^{2}\left(  B\right)
{\textstyle\bigoplus}
L^{2}\left(  B\right)
{\textstyle\bigoplus}
L^{2}\left(  B\right)  .
\]
Then, $\mathfrak{H}(B)$ is a closed subspace of $\mathfrak{H}$, and thus
$\mathfrak{H}(B)$ has an orthogonal complement $\mathfrak{H}^{\bot}(B)$,
i.e.,
\[
\mathfrak{H}=\mathfrak{H}(B)%
{\textstyle\bigoplus}
\mathfrak{H}^{\bot}(B).
\]
We set $\boldsymbol{\Pi}_{B}$ as the projection $\mathfrak{H}\rightarrow
\mathfrak{H}(B)$. Notice that $\boldsymbol{\Pi}_{B}\left(  \psi\right)
=1_{B}\psi$, and that $\boldsymbol{\Pi}_{B}$ is a bounded, self-adjoint
operator, which is an extension of operator $\boldsymbol{E}(B_{m_{\lambda_{1}%
}}\times B_{m_{\lambda_{2}}}\times B_{m_{\lambda_{3}}})$. We use the notation
$\boldsymbol{\Pi}_{B}$\ instead of $\boldsymbol{E}\left(  B\right)  $ to
emphasize that the construction of $\boldsymbol{\Pi}_{B}$ is not based on the
spectral theorem. We interpret $\boldsymbol{\Pi}_{B}$\ as the property of a
system to be localized in $B$. If the state of a system is $\psi
\in\mathfrak{H}$, $\left\Vert \psi\right\Vert =1$, then the probability of
finding the system in state $\boldsymbol{\Pi}_{B}\psi$ localized \ in $B$ is
$\left(  \psi,\boldsymbol{\Pi}_{B}\psi\right)  $.

\begin{lemma}
\label{Lemma1}Set $\phi\left(  x_{i}\right)  :=\wp^{\frac{R_{0}}{2}}%
\Omega\left(  \wp^{R_{0}}\left\vert x_{i}\right\vert _{\wp}\right)  $, where
$R_{0}\in\mathbb{Z}$. With this notation the following assertions hold true:

\noindent(i)
\begin{equation}
\Omega\left(  \wp^{R_{0}}\left\vert x_{i}\right\vert _{\wp}\right)  \psi
_{rnj}\left(  x_{i}\right)  =\left\{
\begin{array}
[c]{ll}%
\psi_{rnj}\left(  x_{i}\right)  & \text{if }n\wp^{-r}\in\wp^{R_{0}}%
\mathbb{Z}_{\wp}\text{,\ }r\leq-R_{0}\\
& \\
\wp^{-\frac{r}{2}}\Omega\left(  \wp^{R_{0}}\left\vert x_{i}\right\vert _{\wp
}\right)  & \text{if }n\wp^{-r}\in\wp^{-r}\mathbb{Z}_{\wp}\text{,\ }%
r\geq-R_{0}+1\\
& \\
0 & \text{if }n\wp^{-r}\notin\wp^{-r}\mathbb{Z}_{\wp}\text{, \ }r\geq-R_{0}+1.
\end{array}
\right.  \label{Restriction}%
\end{equation}

\noindent(ii) The Fourier expansion of $\phi\left(  x_{i}\right)  $ respect to
the basis $\left\{  \psi_{rnj}\left(  x_{i}\right)  \right\}  _{rnj}$ is given
by%
\[
\phi\left(  x_{i}\right)  =\wp^{-\frac{R_{0}}{2}}%
{\displaystyle\sum\limits_{r\geq-R_{0}+1}}
{\displaystyle\sum\limits_{j}}
\wp^{-\frac{r}{2}}\psi_{r0j}\left(  x\right)  ,
\]
where the support of $\psi_{r0j}\left(  x_{i}\right)  $ is $\wp^{-r}%
\mathbb{Z}_{\wp}$.

\noindent(iii) The Fourier expansion of $\Omega\left(  \wp^{R_{0}}\left\Vert
\boldsymbol{x}\right\Vert _{\wp}\right)  \subset\mathbb{Q}_{\wp}^{3}$ in the
basis $\left\{  \psi_{\boldsymbol{rnj}}\left(  \boldsymbol{x}\right)
\right\}  _{\boldsymbol{rnj}}$ is%
\[
\wp^{\frac{3R_{0}}{2}}\Omega\left(  \wp^{R_{0}}\left\Vert \boldsymbol{x}%
\right\Vert _{\wp}\right)  =\wp^{-\frac{3R_{0}}{2}}%
{\displaystyle\sum\limits_{\substack{r_{1}\geq-R_{0}+1\\r_{2}\geq
-R_{0}+1\\r_{3}\geq-R_{0}+1}}}
\text{ }%
{\displaystyle\sum\limits_{j_{1},j_{2},j_{3}}}
\wp^{-\frac{\left(  r_{1}+r_{2}+r_{3}\right)  }{2}}\psi_{\boldsymbol{r0j}%
}\left(  \boldsymbol{x}\right)  ,
\]

where $\boldsymbol{r}=\left(  r_{1},r_{2},r_{3}\right)  $, $\boldsymbol{j}%
=\left(  j_{1},j_{2},j_{3}\right)  $.
\end{lemma}

\begin{proof}
(i) The formula is well-known by the experts. For the sake of completeness, we
review it here. Recall that $\mathrm{supp}$ $\psi_{rnj}=n\wp^{-r}+\wp
^{-r}\mathbb{Z}_{p}$. The cases that appear in (\ref{Restriction}) corresponds
to
\begin{equation}
n\wp^{-r}+\wp^{-r}\mathbb{Z}_{p}\subseteq\wp^{R_{0}}\mathbb{Z}_{p},
\label{Case 1}%
\end{equation}%
\begin{equation}
n\wp^{-r}+\wp^{-r}\mathbb{Z}_{p}\varsupsetneq\wp^{R_{0}}\mathbb{Z}_{p},
\label{Case 2}%
\end{equation}%
\begin{equation}
n\wp^{-r}+\wp^{-r}\mathbb{Z}_{p}\cap\wp^{R_{0}}\mathbb{Z}_{p}=\varnothing.
\label{Case 3}%
\end{equation}
In the first case, $n\wp^{-r}\in\wp^{R_{0}}\mathbb{Z}_{p}$, and thus $\wp
^{-r}\mathbb{Z}_{p}\subseteq\wp^{R_{0}}\mathbb{Z}_{p}$, which equivalent to
$-r\geq R_{0}$. Conversely $-r\geq R_{0}$ (i.e., $\wp^{-r}\mathbb{Z}%
_{p}\subseteq\wp^{R_{0}}\mathbb{Z}_{p}$) and $n\wp^{-r}\in\wp^{R_{0}%
}\mathbb{Z}_{p}$ imply (\ref{Case 1}). In the second case, $0\in n\wp^{-r}%
+\wp^{-r}\mathbb{Z}$, and since any point of a ball is its center, we have
$n\wp^{-r}+\wp^{-r}\mathbb{Z=}\wp^{-r}\mathbb{Z}$, which implies that
$n\wp^{-r}\in\wp^{-r}\mathbb{Z}$. Now, the condition, $\wp^{-r}\mathbb{Z}%
_{p}\varsupsetneq\wp^{R_{0}}\mathbb{Z}_{p}$ is equivalent to $R_{0}>-r$. The
converse assertion can be easily verified. The last case follows from the
first two cases.

(ii) The Fourier expansion of $\phi\left(  x_{i}\right)  $ in the basis
$\left\{  \psi_{rnj}\right\}  _{_{rnj}}$ is given by
\[
\phi\left(  x_{i}\right)  =%
{\displaystyle\sum\limits_{rnj}}
C_{rnj}\psi_{rnj}\left(  x_{i}\right)  \text{, with }C_{rnj}=%
{\displaystyle\int\limits_{\mathbb{Q}_{\wp}}}
\phi\left(  x_{i}\right)  \psi_{rnj}\left(  x_{i}\right)  dx_{i}.
\]
Then, the coefficient $C_{rnj}$ depends on restriction of $\psi_{rnj}\left(
x_{i}\right)  $ to the ball $\wp^{R_{0}}\mathbb{Z}_{\wp}$. The condition
$np^{-r}\in\wp^{-r}\mathbb{Z}_{\wp}$, \ $r\geq-R_{0}+1$ in the second in
(\ref{Restriction}), is equivalent to $n=0$, $r\geq-R_{0}+1$. If the support
of $\psi_{rnj}\left(  x_{i}\right)  \subseteq\wp^{R_{0}}\mathbb{Z}_{p}$, this
case corresponds \ to the first line in (\ref{Restriction}), by using that
\[%
{\displaystyle\int\limits_{\mathbb{Q}_{\wp}}}
\psi_{rnj}\left(  x_{i}\right)  dx_{i}=0,
\]
we have $C_{rnj}=0$. If the support of $\psi_{rnj}\left(  x_{i}\right)
\supsetneqq\wp^{R_{0}}\mathbb{Z}_{p}$,\ this case corresponds to the second
line in (\ref{Restriction}), then
\[
C_{r0j}=\wp^{\frac{R_{0}}{2}}%
{\displaystyle\int\limits_{\mathbb{Q}_{\wp}}}
\wp^{-\frac{r}{2}}\Omega\left(  \wp^{R_{0}}\left\vert x_{i}\right\vert _{\wp
}\right)  dx_{i}=\wp^{-\frac{R_{0}}{2}-\frac{r}{2}}\text{, }r\geq-R_{0}+1.
\]
Therefore,
\[
\phi\left(  x_{i}\right)  =%
{\displaystyle\sum\limits_{r\geq-R_{0}+1}}
{\displaystyle\sum\limits_{k}}
\wp^{-\frac{R_{0}}{2}-\frac{r}{2}}\psi_{r0j}\left(  x_{i}\right)  .
\]

(iii) It follows directly form (ii).
\end{proof}

We set%
\begin{equation}
\boldsymbol{\wp}^{1-\boldsymbol{r}}:=\left(  \wp^{1-r_{1}},\wp^{1-r_{2}}%
,\wp^{1-r_{3}}\right)  \text{, \ \ }\boldsymbol{C}^{T}=\left[
\begin{array}
[c]{cccc}%
C_{1} & C_{2} & C_{3} & C_{4}%
\end{array}
\right]  \in\mathbb{C}^{4}. \label{Notation}%
\end{equation}

\begin{lemma}
\label{Lemma2}(i) With the above notation, it verifies that%
\[
\left(  e^{-i\boldsymbol{H}_{0}t}\mathcal{P}_{\substack{pos\\neg}%
}\boldsymbol{C}\psi_{\boldsymbol{rnj}}\right)  \left(  \boldsymbol{x}\right)
=e^{\mp\lambda\left(  \boldsymbol{\wp}^{1-\boldsymbol{r}}\right)  t}\left(
\frac{\pm\lambda\left(  \boldsymbol{\wp}^{1-\boldsymbol{r}}\right)
+\boldsymbol{h}\left(  \boldsymbol{\wp}^{1-\boldsymbol{r}}\right)  }%
{2\lambda\left(  \boldsymbol{\wp}^{1-\boldsymbol{r}}\right)  }\right)
\boldsymbol{C}\psi_{\boldsymbol{rnj}}\left(  \boldsymbol{x}\right)  .
\]

\noindent(ii) Take $\psi\left(  \boldsymbol{x}\right)  =\sum_{\boldsymbol{rnj}%
}\psi_{\boldsymbol{rnj}}\left(  \boldsymbol{x}\right)  $ and $\boldsymbol{C}$
as in (\ref{Notation}), such that
\[
\psi\left(  \boldsymbol{x}\right)  \boldsymbol{C}\in\mathfrak{H}%
_{\substack{pos\\neg}}.
\]
Then%
\[
\left(  e^{-i\boldsymbol{H}_{0}t}\psi\boldsymbol{C}\right)  \left(
\boldsymbol{x}\right)  =%
{\displaystyle\sum\limits_{\boldsymbol{rnj}}}
e^{\mp\lambda\left(  \boldsymbol{\wp}^{1-\boldsymbol{r}}\right)  t}\left(
\frac{\pm\lambda\left(  \boldsymbol{\wp}^{1-\boldsymbol{r}}\right)
+\boldsymbol{h}\left(  \boldsymbol{\wp}^{1-\boldsymbol{r}}\right)  }%
{2\lambda\left(  \boldsymbol{\wp}^{1-\boldsymbol{r}}\right)  }\right)
\boldsymbol{C}\psi_{\boldsymbol{rnj}}\left(  \boldsymbol{x}\right)
\]

\end{lemma}

\begin{proof}
The Fourier transform of $\psi_{r_{i}n_{i}j_{i}}\left(  x_{i}\right)  $ is
\[
\widehat{\psi}_{r_{i}n_{i}j_{i}}\left(  p_{i}\right)  =\wp^{\frac{r_{i}}{2}%
}\chi_{\wp}\left(  \wp^{-r_{i}}n_{i}p_{i}\right)  \Omega\left(  \left\vert
\wp^{-r_{i}}p_{i}+\wp^{-1}j_{i}\right\vert _{p}\right)  ,
\]
and then%
\[
\widehat{\psi}_{\boldsymbol{rnj}}\left(  \boldsymbol{p}\right)  =%
{\displaystyle\prod\limits_{i=1}^{3}}
\widehat{\psi}_{r_{i}n_{i}j_{i}}\left(  p_{i}\right)  .
\]
Now, for a radial function $\mathfrak{a}\left(  \left\vert \underline
{\boldsymbol{p}}\right\vert _{\wp}\right)  =\mathfrak{a}\left(  \left\vert
p_{1}\right\vert _{\wp},\left\vert p_{2}\right\vert _{\wp},\left\vert
p_{3}\right\vert _{\wp}\right)  $, it verifies that
\begin{align*}
\mathfrak{a}\left(  \left\vert \underline{\boldsymbol{p}}\right\vert _{\wp
}\right)  \widehat{\psi}_{\boldsymbol{rnj}}\left(  \boldsymbol{p}\right)   &
=\mathfrak{a}\left(  \left\vert \wp^{r_{1}-1}j_{1}\right\vert _{\wp
},\left\vert \wp^{r_{2}-1}j_{2}\right\vert _{\wp},\left\vert \wp^{r_{3}%
-1}j_{3}\right\vert _{\wp}\right)  \widehat{\psi}_{\boldsymbol{rnj}}\left(
\boldsymbol{p}\right) \\
&  =\mathfrak{a}\left(  \wp^{1-r_{1}},\wp^{1-r_{2}},\wp^{1-r_{3}}\right)
\widehat{\psi}_{\boldsymbol{rnj}}\left(  \boldsymbol{p}\right)  =\mathfrak{a}%
\left(  \boldsymbol{\wp}^{1-\boldsymbol{r}}\right)  \widehat{\psi
}_{\boldsymbol{rnj}}\left(  \boldsymbol{p}\right)  .
\end{align*}
By using this formulae, we have%
\[
e^{\mp i\lambda\left(  \boldsymbol{p}\right)  t}\frac{1}{2}\left(
\boldsymbol{1}\pm\frac{\boldsymbol{h}\left(  \boldsymbol{p}\right)  }%
{\lambda\left(  \boldsymbol{p}\right)  }\right)  \boldsymbol{C}\widehat{\psi
}_{\boldsymbol{rnj}}\left(  \boldsymbol{p}\right)  =e^{\mp\lambda\left(
\boldsymbol{\wp}^{1-\boldsymbol{r}}\right)  t}\left(  \frac{\pm\lambda\left(
\boldsymbol{\wp}^{1-\boldsymbol{r}}\right)  +\boldsymbol{h}\left(
\boldsymbol{\wp}^{1-\boldsymbol{r}}\right)  }{2\lambda\left(  \boldsymbol{\wp
}^{1-\boldsymbol{r}}\right)  }\right)  \boldsymbol{C}\widehat{\psi
}_{\boldsymbol{rnj}}\left(  \boldsymbol{p}\right)  .
\]
The result follows from this formula by Theorem \ref{Theorem4}.

(ii) It follows from the first part.
\end{proof}

\begin{theorem}
\label{Theorem5}Take $\boldsymbol{a}^{T}=\left[
\begin{array}
[c]{cccc}%
a_{1} & a_{2} & a_{3} & a_{4}%
\end{array}
\right]  \in\mathbb{R}^{4}$, with $a_{i}>0$, $i=1,2,3,4$,
\[
\left\vert \boldsymbol{a}\right\vert =\sqrt{\left\vert a_{1}\right\vert
^{2}+\left\vert a_{2}\right\vert ^{2}+\left\vert a_{3}\right\vert
^{2}+\left\vert a_{4}\right\vert ^{2}}\neq0,
\]
and\ $L\in\mathbb{Z}$, and the normalized state $\psi_{0}^{+}$,$\left\Vert
\psi_{0}^{+}\right\Vert =1$, defined as
\begin{equation}
\psi_{0}^{+}\left(  \boldsymbol{x}\right)  =\mathcal{P}_{pos}\left(  \frac
{\wp^{\frac{3L}{2}}}{\left\vert \boldsymbol{a}\right\vert }\Omega\left(
\wp^{L}\left\Vert \boldsymbol{x}\right\Vert _{\wp}\right)  \left[
\begin{array}
[c]{c}%
a_{1}\\
a_{2}\\
a_{3}\\
a_{4}%
\end{array}
\right]  \right)  . \label{Eq_P_pos}%
\end{equation}
We fix the ball%
\[
B_{l_{0}}^{3}\left(  \wp^{-l_{0}}\boldsymbol{b}\right)  =\wp^{-l_{0}%
}\boldsymbol{b}+\wp^{-l_{0}}\mathbb{Z}_{\wp}^{3},
\]
where $l_{0}\in\mathbb{Z}$ and \ $\boldsymbol{b}=(b_{1},b_{2},b_{3})\in\left(
\mathbb{Q}_{\wp}/\mathbb{Z}_{\wp}\right)  ^{3}$ satisfy
\[
\wp^{-l_{0}}b_{i}\notin\wp^{-l_{0}}\mathbb{Z}_{\wp}\text{, }i=1,2,3\text{, and
}l_{0}\geq-L+1.
\]
The distance between the balls $B_{-L}^{3}=\wp^{L}\mathbb{Z}_{\wp}^{3}$ and
$B_{l_{0}}^{3}\left(  \wp^{-l_{0}}\boldsymbol{b}\right)  $ is $\wp^{^{l_{0}}%
}\left\Vert \boldsymbol{b}\right\Vert _{\wp}>0$. Then%
\[
\left(  e^{-it\boldsymbol{H}_{0}}\psi_{0}^{+},\boldsymbol{\Pi}_{B_{l_{0}}%
^{3}\left(  \wp^{-l_{0}}\boldsymbol{b}\right)  }e^{-it\boldsymbol{H}_{0}}%
\psi_{0}^{+}\right)  >0\text{ for any }t\in\left(  0,\infty\right)  \text{.}%
\]

\end{theorem}

\begin{remark}
The result is also valid if we replace $\mathcal{P}_{pos}$ by $\mathcal{P}%
_{neg}$ in (\ref{Eq_P_pos}). By Lemma \ref{Lemma1}-(i),
\[
B_{-L}^{3}\cap B_{l_{0}}^{3}\left(  \wp^{-l_{0}}\boldsymbol{b}\right)
=\varnothing.
\]
Now take $\boldsymbol{y}\in\mathbb{Z}_{\wp}^{3}$, with $\left\Vert
\boldsymbol{y}\right\Vert _{\wp}>\wp^{-L}$, i.e. $\boldsymbol{y}\notin
B_{-L}^{3}$, by using the ultrametric property of $\left\Vert
\boldsymbol{\cdot}\right\Vert _{\wp}$ we have $\left\Vert \boldsymbol{x}%
-\boldsymbol{y}\right\Vert _{\wp}=\max\left\{  \left\Vert \boldsymbol{x}%
\right\Vert _{\wp},\left\Vert \boldsymbol{y}\right\Vert _{\wp}\right\}
=\left\Vert \boldsymbol{y}\right\Vert _{\wp}$ for any $\boldsymbol{x}\in
B_{-L}^{3}$, then%
\[
dist\left(  B_{-L}^{3},\boldsymbol{y}\right)  =\inf_{\boldsymbol{x}\in
B_{-L}^{3}}\left\Vert \boldsymbol{x}-\boldsymbol{y}\right\Vert _{\wp
}=\left\Vert \boldsymbol{y}\right\Vert _{\wp}\text{.}%
\]
Now the distance between the balls $B_{-L}^{3}$, $B_{l_{0}}^{3}\left(
\wp^{-l_{0}}\boldsymbol{b}\right)  $ is given by%
\begin{align*}
dist\left(  B_{-L}^{3},B_{l_{0}}^{3}\left(  \wp^{-l_{0}}\boldsymbol{b}\right)
\right)   &  =\inf_{\substack{\boldsymbol{x}\in B_{-L}^{3}\\\boldsymbol{y}\in
B_{l_{0}}^{3}\left(  \wp^{-l_{0}}\boldsymbol{b}\right)  }}\left\Vert
\boldsymbol{x}-\boldsymbol{y}\right\Vert _{\wp}=\inf_{\boldsymbol{y}\in
B_{l_{0}}^{3}\left(  \wp^{-l_{0}}\boldsymbol{b}\right)  }\left(
\inf_{\boldsymbol{x}\in B_{-L}^{3}}\left\Vert \boldsymbol{x}-\boldsymbol{y}%
\right\Vert _{\wp}\right) \\
&  =\inf_{\boldsymbol{y}\in B_{l_{0}}^{3}\left(  \wp^{-l_{0}}\boldsymbol{b}%
\right)  }\left(  \left\Vert \boldsymbol{y}\right\Vert _{\wp}\right)
=\left\Vert \wp^{-l_{0}}\boldsymbol{b}\right\Vert _{\wp}=\wp^{^{l_{0}}%
}\left\Vert \boldsymbol{b}\right\Vert _{\wp}\text{.}%
\end{align*}

\end{remark}

\begin{proof}
By Lemmas \ref{Lemma1}-\ref{Lemma2}, with $L=R_{0}$, and
\begin{gather*}
\left(  \frac{\lambda\left(  \boldsymbol{\wp}^{1-\boldsymbol{r}}\right)
+\boldsymbol{h}\left(  \boldsymbol{\wp}^{1-\boldsymbol{r}}\right)  }%
{2\lambda\left(  \boldsymbol{\wp}^{1-\boldsymbol{r}}\right)  }\right)  \left[
\begin{array}
[c]{c}%
a_{1}\\
a_{2}\\
a_{3}\\
a_{4}%
\end{array}
\right] \\
=\frac{1}{2\lambda\left(  \boldsymbol{\wp}^{1-\boldsymbol{r}}\right)  }\left[
\begin{array}
[c]{c}%
a_{1}\lambda\left(  \boldsymbol{\wp}^{1-\boldsymbol{r}}\right)  +c^{2}%
ma_{1}+ca_{3}\wp_{3}^{-r_{3}+1}+a_{4}\left(  c\wp_{1}^{-r_{1}+1}-ic\wp
_{2}^{-r_{2}+1}\right) \\
a_{2}\lambda\left(  \boldsymbol{\wp}^{1-\boldsymbol{r}}\right)  +c^{2}%
ma_{2}-ca_{4}\wp_{3}^{-r_{3}+1}+a_{3}\left(  c\wp_{1}^{-r_{1}+1}+ic\wp
_{2}^{-r_{2}+1}\right) \\
a_{3}\lambda\left(  \boldsymbol{\wp}^{1-\boldsymbol{r}}\right)  -c^{2}%
ma_{3}+ca_{1}\wp_{3}^{-r_{3}+1}+a_{2}\left(  c\wp_{1}^{-r_{1}+1}-ic\wp
_{2}^{-r_{2}+1}\right) \\
a_{4}\lambda\left(  \boldsymbol{\wp}^{1-\boldsymbol{r}}\right)  -c^{2}%
ma_{4}-ca_{2}\wp_{3}^{-r_{3}+1}+a_{1}\left(  c\wp_{1}^{-r_{1}+1}+ic\wp
_{2}^{-r_{2}+1}\right)
\end{array}
\right] \\
=:\left[
\begin{array}
[c]{c}%
A_{1}\left(  \boldsymbol{\wp}^{1-\boldsymbol{r}}\right) \\
A_{2}\left(  \boldsymbol{\wp}^{1-\boldsymbol{r}}\right) \\
A_{3}\left(  \boldsymbol{\wp}^{1-\boldsymbol{r}}\right) \\
A_{4}\left(  \boldsymbol{\wp}^{1-\boldsymbol{r}}\right)
\end{array}
\right]  ,
\end{gather*}
we have%
\[
e^{-it\boldsymbol{H}_{0}}\psi_{0}^{+}\left(  \boldsymbol{x}\right)
=\wp^{-\frac{3L}{2}}%
{\displaystyle\sum\limits_{\substack{r_{1}\geq-L+1\\r_{2}\geq-L+1\\r_{3}%
\geq-L+1}}}
\text{ }%
{\displaystyle\sum\limits_{j_{1},j_{2},j_{3}}}
\wp^{-\frac{\left(  r_{1}+r_{2}+r_{3}\right)  }{2}}e^{-\lambda\left(
\boldsymbol{\wp}^{1-\boldsymbol{r}}\right)  t}\psi_{\boldsymbol{r0j}}\left(
\boldsymbol{x}\right)  \left[
\begin{array}
[c]{c}%
A_{1}\left(  \boldsymbol{\wp}^{1-\boldsymbol{r}}\right) \\
A_{2}\left(  \boldsymbol{\wp}^{1-\boldsymbol{r}}\right) \\
A_{3}\left(  \boldsymbol{\wp}^{1-\boldsymbol{r}}\right) \\
A_{4}\left(  \boldsymbol{\wp}^{1-\boldsymbol{r}}\right)
\end{array}
\right]  ,
\]
where \ $\boldsymbol{r}=\left(  r_{1},r_{2},r_{3}\right)  $. We now compute
\begin{multline*}
\boldsymbol{\Pi}_{B_{l_{0}}^{3}\left(  \wp^{-l_{0}}\boldsymbol{b}\right)
}e^{-it\boldsymbol{H}_{0}}\psi_{0}^{+}=\\
\wp^{-\frac{3L}{2}}%
{\displaystyle\sum\limits_{\substack{r_{1}\geq-L+1\\r_{2}\geq-L+1\\r_{3}%
\geq-L+1}}}
\text{ }%
{\displaystyle\sum\limits_{j_{1},j_{2},j_{3}}}
\wp^{-\frac{\left(  r_{1}+r_{2}+r_{3}\right)  }{2}}e^{-\lambda\left(
\boldsymbol{\wp}^{1-\boldsymbol{r}}\right)  t}\Omega\left(  \left\Vert
\wp^{l_{0}}\boldsymbol{x}-\boldsymbol{b}\right\Vert _{\wp}\right)
\psi_{\boldsymbol{r0j}}\left(  \boldsymbol{x}\right)  \left[
\begin{array}
[c]{c}%
A_{1}\left(  \boldsymbol{\wp}^{1-\boldsymbol{r}}\right) \\
A_{2}\left(  \boldsymbol{\wp}^{1-\boldsymbol{r}}\right) \\
A_{3}\left(  \boldsymbol{\wp}^{1-\boldsymbol{r}}\right) \\
A_{4}\left(  \boldsymbol{\wp}^{1-\boldsymbol{r}}\right)
\end{array}
\right]
\end{multline*}
Now, since
\[
\mathrm{supp}\text{ }\psi_{\boldsymbol{r0j}}\left(  \boldsymbol{x}\right)
=\wp^{-r_{1}}\mathbb{Z}_{\wp}\times\wp^{-r_{2}}\mathbb{Z}_{\wp}\times
\wp^{-r_{3}}\mathbb{Z}_{\wp},\text{ }%
\]
by taking $l_{0}\geq-r_{j}$, since $ord(b_{j})\geq0$, we have $-l_{0}%
+ord(b_{j})\geq-r_{j}$ and%
\[
\wp^{-l_{0}}b_{j}+\wp^{-l_{0}}\mathbb{Z}_{\wp}\subseteq\wp^{-r_{j}}%
\mathbb{Z}_{\wp}^{3},\text{ for }j=1,2,3\text{.}%
\]
Therefore
\[
\wp^{-l_{0}}\boldsymbol{b}+\wp^{-l_{0}}\mathbb{Z}_{\wp}^{3}\subseteq
\mathrm{supp}\text{ }\psi_{\boldsymbol{r0j}}\left(  \boldsymbol{x}\right)
\text{ for }r_{j}\geq-\min\left\{  L-1,l_{0}\right\}  \text{, }j=1,2,3\text{,}%
\]
and
\begin{multline*}
\boldsymbol{\Pi}_{B_{l_{0}}^{3}\left(  \wp^{-l_{0}}\boldsymbol{b}\right)
}e^{-it\boldsymbol{H}_{0}}\psi_{0}^{+}=\\
\wp^{-\frac{3L}{2}}%
{\displaystyle\sum\limits_{\substack{r_{1}\geq-\min\left\{  L-1,l_{0}\right\}
\\r_{2}\geq-\min\left\{  L-1,l_{0}\right\}  \\r_{3}\geq-\min\left\{
L-1,l_{0}\right\}  }}}
\text{ }%
{\displaystyle\sum\limits_{j_{1},j_{2},j_{3}}}
\wp^{-\frac{\left(  r_{1}+r_{2}+r_{3}\right)  }{2}}e^{-\lambda\left(
\boldsymbol{\wp}^{1-\boldsymbol{r}}\right)  t}\psi_{\boldsymbol{r0j}}\left(
\boldsymbol{x}\right)  \left[
\begin{array}
[c]{c}%
A_{1}\left(  \boldsymbol{\wp}^{1-\boldsymbol{r}}\right) \\
A_{2}\left(  \boldsymbol{\wp}^{1-\boldsymbol{r}}\right) \\
A_{3}\left(  \boldsymbol{\wp}^{1-\boldsymbol{r}}\right) \\
A_{4}\left(  \boldsymbol{\wp}^{1-\boldsymbol{r}}\right)
\end{array}
\right]  .
\end{multline*}

Finally,%
\begin{multline*}
\left(  e^{-it\boldsymbol{H}_{0}}\psi_{0}^{+},\boldsymbol{\Pi}_{B_{l_{0}}%
^{3}\left(  \wp^{-l_{0}}\boldsymbol{b}\right)  }e^{-it\boldsymbol{H}_{0}}%
\psi_{0}^{+}\right)  =\\%
{\displaystyle\sum\limits_{k=1}^{4}}
\wp^{-\frac{3L}{2}}%
{\displaystyle\sum\limits_{\substack{r_{1}\geq-\min\left\{  L-1,l_{0}\right\}
\\r_{2}\geq-\min\left\{  L-1,l_{0}\right\}  \\r_{3}\geq-\min\left\{
L-1,l_{0}\right\}  }}}
\text{ }%
{\displaystyle\sum\limits_{j_{1},j_{2},j_{3}}}
\wp^{-\frac{\left(  r_{1}+r_{2}+r_{3}\right)  }{2}}e^{-\lambda\left(
\boldsymbol{\wp}^{1-\boldsymbol{r}}\right)  t}%
{\displaystyle\int\limits_{\mathbb{Q}_{\wp}^{3}}}
\left\vert \psi_{\boldsymbol{r0j}}\left(  \boldsymbol{x}\right)  \right\vert
^{2}\left\vert A_{k}\left(  \boldsymbol{\wp}^{1-\boldsymbol{r}}\right)
\right\vert ^{2}d^{3}\boldsymbol{x}\\
=\wp^{-\frac{3L}{2}}%
{\displaystyle\sum\limits_{k=1}^{4}}
{\displaystyle\sum\limits_{\substack{r_{1}\geq-\min\left\{  L-1,l_{0}\right\}
\\r_{2}\geq-\min\left\{  L-1,l_{0}\right\}  \\r_{3}\geq-\min\left\{
L-1,l_{0}\right\}  }}}
\text{ }%
{\displaystyle\sum\limits_{j_{1},j_{2},j_{3}}}
\wp^{-\frac{\left(  r_{1}+r_{2}+r_{3}\right)  }{2}}e^{-\lambda\left(
\boldsymbol{\wp}^{1-\boldsymbol{r}}\right)  t}\left\vert A_{k}\left(
\boldsymbol{\wp}^{1-\boldsymbol{r}}\right)  \right\vert ^{2}>0,
\end{multline*}
for any $t\in\left(  0,\infty\right)  $.
\end{proof}

The Einstein causality requires a finite propagation speed for all physical
particles. In the standard case, any solution of the Dirac equation propagates
slower than the speed of light. This requires that the support of any state
$\psi\in\mathfrak{H}_{pos}\cup\mathfrak{H}_{neg}$ be the whole $\mathbb{R}%
^{3}$; see \cite[Section 1.8.2]{Thaller}. By Theorem \ref{Theorem5}, the
transition probability from a localized state in a ball $B_{-L}^{3}$ to a
state localized in a ball $B_{l_{0}}^{3}\left(  \wp^{-l_{0}}\boldsymbol{b}%
\right)  $, which is arbitrary far away from $B_{-L}^{3}$, is positive for any
time $t\in\left(  0,\epsilon\right)  $, where $\epsilon$ is arbitrarily small.
Then, the system has a non-zero probability of getting from $B_{-L}^{3}$\ to
$B_{l_{0}}^{3}\left(  \wp^{-l_{0}}\boldsymbol{b}\right)  $ in an arbitrary
short time, thereby propagation with superluminal speed. This feature is a
consequence of the discrete nature of the space $\mathbb{Q}_{\wp}^{3}$.

\section{Conclusions}

Although the validity of the Lorentz symmetry has been experimentally proven
with great precision \cite{Kostelecky-Russell}, there is a consensus within
the community of quantum-spacetime phenomenology (particularly in the
quantum-gravity community) that the breaking of Lorentz symmetry occurs at the
Planck scale, \cite{Amelino-Camelia}. The breakdown of this symmetry has been
investigated in connection with many other physical phenomena, \cite{Mariz et
al}.

The study of the limit of quantum mechanics at the Planck length is a central
scientific problem connected to the unification of general relativity and
quantum mechanics. Since 1990, it has been recognized that the Planck length
has substantial implications in quantum mechanics. Among them, the
localization of a particle with better accuracy than its Compton wavelength is
impossible. Then, it is necessary to abandon the notion of particles in favor
of localized particles. Also, it was recognized that the discreteness of
space-time may imply the possibility that particles would travel faster than
light. However, it was pointed out that this paradox disappears when allowing
the creation and annihilation of particles, \cite{Garay}.

The $\wp$-adic Dirac equation is the simplest model where the abovementioned
matters can be discussed in a precise mathematical form. The new model uses
$\mathbb{R}\times\mathbb{Q}_{\wp}^{3}$ as space-time, with $t\in\mathbb{R}$
and $\boldsymbol{x}\in\mathbb{Q}_{\wp}^{3}$, and thus the Lorentz symmetry is
naturally broken. The space $\mathbb{Q}_{\wp}^{3}$ is a completely
disconnected topological space, which means that the\ only connected subsets
are points or the empty set, i.e., the there are no `intervals.' This space is
a self-similar set admitting a group of symmetries of the form $\boldsymbol{x}%
\rightarrow\boldsymbol{a}+\boldsymbol{Ax}$, $\boldsymbol{a}\in\mathbb{Q}_{\wp
}^{3}$, $\boldsymbol{A}\in GL_{3}(\mathbb{Q}_{\wp}^{3})$. The action of this
group on the space $\mathbb{Q}_{\wp}^{3}$ imposes a Planck length of exactly
$\wp^{-1}$, which is independent of the speed of light. In contrast, the
classical model does not have a Planck length because, in $\mathbb{R}$, the
Archimedean axiom implies the existence of arbitrary small segments.

The two types of Dirac equations have common properties; particularly, both
predict the existence of pairs of particle and antiparticles. However, in the
$\wp$-adic case, the geometry of $\mathbb{Q}_{\wp}^{3}$ allows localized
states (particles) to exist. In contrast, this possibility is ruled out in the
standard case since it implies the violation of Einstein's causality. Our main
result states the violation of the Einstein causality in the space-time
$\mathbb{R}\times\mathbb{Q}_{\wp}^{3}$ in the Dirac-von Neumann formalism of
quantum mechanics.

The $\wp$-adic nature of space-time was conjectured by Volovich in 1980s,
\cite{Volovich}. Quantum mechanical theories with $\wp$-adic space and time
are possible, see, e.g. \cite{Khrennikov-QM}-\cite{Khrennikov-QM-2},
\cite{Zuniga-RMP}-\cite{Zuniga-RMP-2},\ and the references therein. The
assumption that the time is $\wp$-adic requires abandoning the QM in the sense
of Dirac-von Neumann because the evolution of quantum states is based on the
theory of semigroups, which uses real-time. This, in turn implies that we
cannot compute transition probabilities in classical way. Since the study of
the Einstein causality is based on computing transition probabilities, here
we\ do not use a $\wp$-adic time.

Finally, this work indicates the difficulty of unifying QM and gravity.
Assuming as space-time $\mathbb{R}\times\mathbb{R}^{3}$, QM and general
relativity together imply that
\begin{equation}
L_{\text{Planck}}=\sqrt{\frac{\hbar G}{c^{3}}}\text{ and Einstein causality
principle.} \label{EP}%
\end{equation}
Interpreting the Bronstein inequality $\Delta x\geq L_{\text{Planck}}$ as the
non-existence of intervals below the Planck length and assuming as space-time
$\mathbb{R}\times\mathbb{Q}_{\wp}^{3}$, QM implies that%
\begin{equation}
L_{\text{Planck}}=\wp^{-1}\text{ and the violation of Einstein causality
principle.} \label{NEP}%
\end{equation}
Now, due to the non-existence of topological and algebraic isomorphism between
$\mathbb{R}$ and $\mathbb{Q}_{\wp}$, the conclusions (\ref{EP})-(\ref{NEP})
cannot `mixed'. This fact evokes the old idea of the adelic nature of space,
\cite{V-V-Z}.

\textbf{Conflict of Interests/Competing Interests}

I have no conflicts of interest to disclose.

\end{document}